\documentclass{elsarticle}

\makeatletter
\def\ps@pprintTitle{%
  \let\@oddhead\@empty
  \let\@evenhead\@empty
  \let\@oddfoot\@empty
  \let\@evenfoot\@oddfoot
}
\makeatother  %surpresses "Preprint submitted to Elsevier"

\usepackage{amsmath,enumerate,color}
\usepackage{amssymb,graphicx}
\usepackage{mathdots}
\usepackage[tableposition=above]{caption}
\usepackage{subcaption}
\usepackage{float}
\usepackage{nicefrac}
\usepackage[utf8]{inputenc}
\usepackage{centernot}
\usepackage[]{todonotes} %add option [disable] to hide all orange boxes
\usepackage{multirow}
\usepackage{multicol}
\usepackage{etoolbox}
\usepackage{algorithm2e}
\usepackage{mathtools}
\usepackage{comment}
\usepackage[normalem]{ulem}

\providetoggle{long}
\settoggle{long}{true}
\usepackage{tikz}
\usetikzlibrary{arrows}
\usetikzlibrary{shapes}
\tikzstyle{every path}=[thick]

\long\def\Omit#1{}
\def\diss{{\delta}}

\usepackage{float}
\usepackage{microtype}

\definecolor{darkgreen}{rgb}{0.0, 0.5, 0.0}

\newdefinition{definition}{Definition}

\newtheorem{theorem}{Theorem}
\newtheorem{lemma}[theorem]{Lemma}
\newtheorem{coro}[theorem]{Corollary}
\newtheorem{remark}[theorem]{Remark}
\newproof{proof}{Proof}

\usepackage{tabularx, environ}

\makeatletter

\newcolumntype{\expand}{}
\long\@namedef{NC@rewrite@\string\expand}{\expandafter\NC@find}

\NewEnviron{fancyproblem}[2][]{%
	\def\problem@arg{#1}%
	\def\problem@framed{framed}%
	\def\problem@lined{lined}%
	\def\problem@doublelined{doublelined}%
	\ifx\problem@arg\@empty%
	\def\problem@hline{}%
	\else%
	\ifx\problem@arg\problem@doublelined%
	\def\problem@hline{\hline\hline}%
	\else%
	\def\problem@hline{\hline}%
	\fi%
	\fi%
	\ifx\problem@arg\problem@framed%
	\def\problem@tablelayout{|>{\bfseries}lX|c}%
	\def\problem@title{\multicolumn{2}{|l|}{%
			\raisebox{-\fboxsep}{\textsc{\large #2}}%
	}}%
	\else
	\def\problem@tablelayout{>{\bfseries}lXc}%
	\def\problem@title{\multicolumn{2}{l}{%
			\raisebox{-\fboxsep}{\textsc{\large #2}}%
	}}%
	\fi%
	\bigskip\par\noindent%
	\renewcommand{\arraystretch}{1.2}%
	\begin{tabularx}{\textwidth}{\expand\problem@tablelayout}%
		\problem@hline%
		\problem@title\\[2\fboxsep]%
		\BODY\\\problem@hline%
	\end{tabularx}%
	\medskip\par%
}
\makeatother
\colorlet{darkgreen}{green!60!black}
 	
\def\real{\hbox{\rm\vrule\kern-1pt R}}
\def\nat{\hbox{\rm\vrule\kern-1pt N}}

\newcommand{\MINMAX}{\textsc{Min-Max Dissatisfaction}}
\newcommand{\MINSUM}{\textsc{Min-Sum  Dissatisfaction}}
\newcommand{\MINMAXD}{\textsc{Min-Max Diss.}}
\newcommand{\MINSUMD}{\textsc{Min-Sum  Diss.}}
\newcommand{\newAD}[1]{\textcolor{black}{#1}}
\newcommand{\newSL}[1]{\textcolor{black}{#1}}
\newcommand{\newwSL}[1]{\textcolor{black}{#1}}

\newcommand{\newMM}[1]{\textcolor{black}{#1}}

\title{Allocation of Indivisible Items\\with Individual Preference Graphs}

 \newcommand{\pictureoutstar}{
\begin{tikzpicture}[yscale=0.75]
\def\x{1}
\begin{scope}[scale=0.7]
\node (G1) at (-2,0.5) {$G_1$};
\node[fill=white,draw,circle,inner sep=0pt,minimum size=2mm,label=above:{$\mathstrut a$}] (va)  at (0,\x)   {};
\node[fill=white,draw,circle,inner sep=0pt,minimum size=2mm,label=above:{$\mathstrut  b$}] (vb) at (3,\x)   {};
\node[fill=white,draw,circle,inner sep=0pt,minimum size=2mm,label=above:{$\mathstrut  c$}] (vc) at (6,\x)   {};
\node[fill=white,draw,circle,inner sep=0pt,minimum size=2mm,label=below:{$\mathstrut d$}] (vd) at (2,0)   {};
\node[fill=white,draw,circle,inner sep=0pt,minimum size=2mm,label=below:{$\mathstrut e$}] (ve) at (-1,0)   {};
\node[fill=white,draw,circle,inner sep=0pt,minimum size=2mm,label=below:{$\mathstrut f$}] (vf) at (5,0)   {};
\node[fill=white,draw,circle,inner sep=0pt,minimum size=2mm,label=below:{$\mathstrut g$}] (vg) at (3,0)   {};
\node[fill=white,draw,circle,inner sep=0pt,minimum size=2mm,label=below:{$\mathstrut h$}] (vh) at (4,0)   {};
\node[fill=white,draw,circle,inner sep=0pt,minimum size=2mm,label=below:{$\mathstrut i$}] (vi) at (6,0)   {};
\node[fill=white,draw,circle,inner sep=0pt,minimum size=2mm,label=above:{$\mathstrut j$}] (vj) at (9,\x)   {};
\node[fill=white,draw,circle,inner sep=0pt,minimum size=2mm,label=above:{$\mathstrut k$}] (vk) at (11,\x)   {};
\node[fill=white,draw,circle,inner sep=0pt,minimum size=2mm,label=above:{$\mathstrut l$}] (vl) at (13,\x)   {};
\node[fill=white,draw,circle,inner sep=0pt,minimum size=2mm,label=below:{$\mathstrut m$}] (vm) at (9,0)   {};
\node[fill=white,draw,circle,inner sep=0pt,minimum size=2mm,label=below:{$\mathstrut n$}] (vn) at (10,0)   {};
\node[fill=white,draw,circle,inner sep=0pt,minimum size=2mm,label=below:{$\mathstrut o$}] (vo) at (0,0)   {};
\node[fill=white,draw,circle,inner sep=0pt,minimum size=2mm,label=below:{$\mathstrut p$}] (vp) at (11,0)   {};
\node[fill=white,draw,circle,inner sep=0pt,minimum size=2mm,label=below:{$\mathstrut q$}] (vq) at (13,0)   {};
\node[fill=white,draw,circle,inner sep=0pt,minimum size=2mm,label=above:{$\mathstrut r$}] (vr) at (14,\x)   {};
\node[fill=white,draw,circle,inner sep=0pt,minimum size=2mm,label=below:{$\mathstrut u$}] (vu) at (7,0)   {};
\draw[->] (va)--(ve);\draw[->] (va)--(vo);
\draw[->] (vb)--(vd);\draw[->] (vb)--(vg);\draw[->] (vb)--(vh);
\draw[->] (vc)--(vf);\draw[->] (vc)--(vi);\draw[->] (vc)--(vu);
\draw[->] (vj)--(vm);
\draw[->] (vk)--(vn);\draw[->] (vk)--(vp);
\draw[->] (vl)--(vq);
\end{scope}

\begin{scope}[scale=0.7,yshift=-3.5cm]
\node (G2) at (-2,0.5) {$G_2$};
\node[fill=white,draw,circle,inner sep=0pt,minimum size=2mm,label=above:{$\mathstrut a$}] (ua)  at (0,\x)   {};
\node[fill=white,draw,circle,inner sep=0pt,minimum size=2mm,label=above:{$\mathstrut b$}] (ub) at (3,\x)   {};
\node[fill=white,draw,circle,inner sep=0pt,minimum size=2mm,label=above:{$\mathstrut c$}] (uc) at (6,\x)   {};
\node[fill=white,draw,circle,inner sep=0pt,minimum size=2mm,label=below:{$\mathstrut d$}] (ud) at (0,0)   {};
\node[fill=white,draw,circle,inner sep=0pt,minimum size=2mm,label=below:{$\mathstrut e$}] (ue) at (-1,0)   {};
\node[fill=white,draw,circle,inner sep=0pt,minimum size=2mm,label=below:{$\mathstrut f$}] (uf) at (3,0)   {};
\node[fill=white,draw,circle,inner sep=0pt,minimum size=2mm,label=below:{$\mathstrut g$}] (ug) at (6,0)   {};
\node[fill=white,draw,circle,inner sep=0pt,minimum size=2mm,label=below:{$\mathstrut h$}] (uh) at (1,0)   {};
\node[fill=white,draw,circle,inner sep=0pt,minimum size=2mm,label=below:{$\mathstrut i$}] (ui) at (7,0)   {};
\node[fill=white,draw,circle,inner sep=0pt,minimum size=2mm,label=below:{$\mathstrut j$}] (uj) at (8,0)   {};
\node[fill=white,draw,circle,inner sep=0pt,minimum size=2mm,label=below:{$\mathstrut k$}] (uk) at (9,0)   {};
\node[fill=white,draw,circle,inner sep=0pt,minimum size=2mm,label=below:{$\mathstrut l$}] (ul) at (5,0)   {};
\node[fill=white,draw,circle,inner sep=0pt,minimum size=2mm,label=above:{$\mathstrut m$}] (um) at (9,\x)   {};
\node[fill=white,draw,circle,inner sep=0pt,minimum size=2mm,label=above:{$\mathstrut n$}] (un) at (11,\x)   {};
\node[fill=white,draw,circle,inner sep=0pt,minimum size=2mm,label=below:{$\mathstrut o$}] (uo) at (11,0)   {};
\node[fill=white,draw,circle,inner sep=0pt,minimum size=2mm,label=below:{$\mathstrut s$}] (us) at (12,0)   {};
\node[fill=white,draw,circle,inner sep=0pt,minimum size=2mm,label=above:{$\mathstrut t$}] (ut) at (13,\x)   {};
\draw[->] (ua)--(ue);\draw[->] (ua)--(ud);\draw[->] (ua)--(uh);
\draw[->] (ub)--(uf);
\draw[->] (uc)--(ul);\draw[->] (uc)--(ug);\draw[->] (uc)--(ui);
\draw[->] (um)--(uj);\draw[->] (um)--(uk);
\draw[->] (un)--(uo);\draw[->] (un)--(us);
\end{scope}
\begin{scope}[yshift=-5cm]

\begin{scope}[scale=0.7]
\node (G) at (-2,-1) {$G'$};
\node[fill=white,draw,circle,inner sep=0pt,minimum size=2mm,label=above:{$2$}] (v1)  at (0,\x)   {};
\node[fill=white,draw,circle,inner sep=0pt,minimum size=2mm,label=above:{$4$}] (v2) at (3,\x)   {};
\node[fill=white,draw,circle,inner sep=0pt,minimum size=2mm,label=above:{$3$}] (v3) at (6,\x)   {};
\node[fill=white,draw,circle,inner sep=0pt,minimum size=2mm,label=below:{$1$}] (v4) at (2,0)   {};
\node[fill=white,draw,circle,inner sep=0pt,minimum size=2mm,label=below:{$1$}] (v5) at (-1,0)   {};
\node[fill=white,draw,circle,inner sep=0pt,minimum size=2mm,label=below:{$1$}] (v6) at (5,0)   {};
\node[fill=white,draw,circle,inner sep=0pt,minimum size=2mm,label=below:{$1$}] (v7) at (3,0)   {};
\node[fill=white,draw,circle,inner sep=0pt,minimum size=2mm,label=below:{$1$}] (v8) at (4,0)   {};
\node[fill=white,draw,circle,inner sep=0pt,minimum size=2mm,label=below:{$1$}] (v9) at (6,0)   {};
\node[fill=white,draw, inner sep=0pt,minimum size=2mm,opacity=0.2,label={[opacity=0.2]above:{}}] (v10) at (9,\x)   {};
\node[fill=white,draw, inner sep=0pt,minimum size=2mm,opacity=0.2,label={[opacity=0.2]above:{}}] (v11) at (11,\x)   {};
\node[fill=white,draw, inner sep=0pt,minimum size=2mm,opacity=0.2,label={[opacity=0.2]above:{}}] (v12) at (13,\x)   {};
\node[fill=white,draw,circle,inner sep=0pt,minimum size=2mm,opacity=0.2,label={[opacity=0.2]below:{}}] (v13) at (9,0)   {};
\node[fill=white,draw,circle,inner sep=0pt,minimum size=2mm,opacity=0.2,label={[opacity=0.2]below:{}}] (v14) at (10,0)   {};
\node[fill=white,draw, inner sep=0pt,minimum size=2mm,opacity=0.2,label=below:{}] (v15) at (0,0)   {};
\node[fill=white,draw,inner sep=0pt,minimum size=2mm,opacity=0.2,label={[opacity=0.2]below:{}}] (v16) at (11,0)   {};
\node[fill=white,draw,inner sep=0pt,minimum size=2mm,opacity=0.2,label={[opacity=0.2]below:{}}] (v17) at (13,0)   {};
\node[fill=white,draw, inner sep=0pt,minimum size=2mm,opacity=0.2,label={[opacity=0.2]above:{}}] (v18) at (14,\x)   {};
\node[fill=white,draw,inner sep=0pt,minimum size=2mm,opacity=0.2,label={[opacity=0.2]below:{}}] (vu) at (7,0)   {};
\draw (v1)--(v5);\draw[opacity=0.2] (v1)--(v15);
\draw (v2)--(v4);\draw (v2)--(v7);\draw (v2)--(v8);
\draw (v3)--(v6);\draw (v3)--(v9);\draw[opacity=0.2] (v3)--(vu);
\draw[opacity=0.2] (v10)--(v13);
\draw[opacity=0.2] (v11)--(v14);
\draw[opacity=0.2] (v11)--(v16);
\draw[opacity=0.2] (v12)--(v17);
\end{scope}

\begin{scope}[scale=0.7,yshift=-3cm]
\node[fill=white,draw,circle,inner sep=0pt,minimum size=2mm,label=above:{$4$}] (u1)  at (0,\x)   {};
\node[fill=white,draw,circle,inner sep=0pt,minimum size=2mm,label=above:{$2$}] (u2) at (3,\x)   {};
\node[fill=white,draw,circle,inner sep=0pt,minimum size=2mm,label=above:{$4$}] (u3) at (6,\x)   {};
\node[fill=white,draw,circle,inner sep=0pt,minimum size=2mm,label=below:{$1$}] (u4) at (0,0)   {};
\node[fill=white,draw,circle,inner sep=0pt,minimum size=2mm,label=below:{$1$}] (u5) at (-1,0)   {};
\node[fill=white,draw,circle,inner sep=0pt,minimum size=2mm,label=below:{$1$}] (u6) at (3,0)   {};
\node[fill=white,draw,circle,inner sep=0pt,minimum size=2mm,label=below:{$1$}] (u7) at (6,0)   {};
\node[fill=white,draw,circle,inner sep=0pt,minimum size=2mm,label=below:{$1$}] (u8) at (1,0)   {};
\node[fill=white,draw,circle,inner sep=0pt,minimum size=2mm,label=below:{$1$}] (u9) at (7,0)   {};
\node[fill=white,draw,circle,inner sep=0pt,minimum size=2mm,opacity=0.2] (u10) at (8,0)   {};
\node[fill=white,draw,circle,inner sep=0pt,minimum size=2mm,opacity=0.2] (u11) at (9,0)   {};
\node[fill=white,draw,circle,inner sep=0pt,minimum size=2mm,opacity=0.2] (u12) at (5,0)   {};
\node[fill=white,draw, inner sep=0pt,minimum size=2mm,opacity=0.2] (u13) at (9,\x)   {};
\node[fill=white,draw, inner sep=0pt,minimum size=2mm,opacity=0.2] (u14) at (11,\x)   {};
\node[fill=white,draw,circle,inner sep=0pt,minimum size=2mm,opacity=0.2,label=below:{}] (u15) at (11,0)   {};
\node[fill=white,draw, inner sep=0pt,minimum size=2mm,opacity=0.2] (u19) at (12,0)   {};
\node[fill=white,draw, inner sep=0pt,minimum size=2mm,opacity=0.2] (u20) at (13,\x)   {};
\draw (u1)--(u5);\draw (u1)--(u4);\draw (u1)--(u8);
\draw (u2)--(u6);
\draw[opacity=0.2] (u3)--(u12);\draw (u3)--(u7);\draw (u3)--(u9);
\draw[opacity=0.2] (u13)--(u10);\draw[opacity=0.2,] (u13)--(u11);
\draw[opacity=0.2] (u14)--(u15);\draw[opacity=0.2] (u14)--(u19);
\end{scope}
\draw[blue,dashed](v1) to[out=-125,in=125] (u1);
\draw[blue,dashed](v2) to[out=-125,in=125] (u2);
\draw[blue,dashed](v3) to[out=-125,in=125] (u3);
\draw[blue,dashed](v4) -- (u4);
\draw[blue,dashed](v5) to[out=-110,in=110] (u5);
\draw[blue,dashed](v6) -- (u6);
\draw[blue,dashed](v7) -- (u7);
\draw[blue,dashed](v8) -- (u8);
\draw[blue,dashed](v9) -- (u9);
\end{scope}

\end{tikzpicture}
}

\begin{document}

\author[1,2]{Nina Chiarelli}
\author[1,2]{Cl\'ement Dallard}
\author[3]{Andreas Darmann}
\author[3]{Stefan Lendl}
\author[1,2]{Martin Milani\v c}
\author[1]{Peter Mur\v si\v c}
\author[3]{Ulrich Pferschy}
\author[1,2]{Nevena Piva\v c}

\address[1]{FAMNIT, University of Primorska, Glagolja\v ska 8, 6000 Koper, Slovenia}
\address[2]{IAM, University of Primorska, Muzejski trg 2, 6000 Koper, Slovenia}
\address[3]{Department of Operations and Information Systems, University of Graz,\\ Universitaetsstrasse 15/E3, 8010 Graz, Austria}

\begin{abstract}
This paper studies the allocation of indivisible items to agents, when each agent's preferences are expressed by means of a directed acyclic graph. 
The vertices of each preference graph represent the subset of items approved of by the respective agent.
An arc $(a,b)$ in such a graph means that the respective agent prefers item $a$ over item $b$.
We introduce a new measure of dissatisfaction of an agent by counting the number of non-assigned items which are approved of by the agent and for which no more preferred item is allocated to the agent.
Considering two problem variants, we seek an allocation of the items to the agents in a way that minimizes (i) the total dissatisfaction over all agents or (ii) the maximum dissatisfaction among the agents. 
For both optimization problems we study the status of computational complexity and obtain NP-hardness results as well as polynomial algorithms with respect to natural underlying graph structures, such as stars, trees, paths, and matchings.
We also analyze the parameterized complexity of the two problems with respect to various parameters related to the number of agents, 
% arborization, 
the dissatisfaction threshold, the vertex degrees of the preference graphs, and the treewidth.
\end{abstract}


\maketitle

\section{Introduction}

Fairly dividing sets of indivisible objects among a set of agents has been studied in the literature from different perspectives (for surveys see, e.g., Bouveret et al.~\cite{bouveret-survey} and Thomson~\cite{thomson}). 
In particular, one can find various models of preferences expressed by the agents and different objectives arising from these. 
In our contribution we introduce a new model of deriving (satisfaction and) dissatisfaction over sets of items based on directed graphs representing partial orders between the items.
This model allows the agents to express also partial or inconclusive preferences in a simple way.
For a layperson such a model should be easier to apply than existing methods where every agent is, e.g., forced to come up with a total order of all available items (or even all subsets of items), to give a full table of all pairwise comparisons, or to assign points in some formal scoring scheme.

As an illustration of our preference model let us consider the situation in which a set of indivisible presents should be divided among a set of kids. 
The kids may be overwhelmed with the task of comparing all available presents among each other, but they are able to state certain preferences such as disapproval of certain presents or strict preference of a certain present over another present. 
A kid will have difficulties to keep an overview of the complicated preference structure resulting from these pairwise comparisons, but it is well versed in complaining when it sees a present given to another kid and it receives no present it likes better than that present.

In such a scenario, the parents want to allocate the presents to the kids in a way that minimizes the (total or maximum) dissatisfaction. 
The dissatisfaction of a kid is here measured by the number of desired presents {\bf not received} and for which the kid does not get any other more preferred present. 
Note that, in this setting, adding less preferred presents will not improve the happiness of a kid. 
This can occur in a more general situation of preferences implied by skills or abilities, where the effect of an object with certain skills is not improved by adding an object with lesser skills.

We introduce a model for such a setting in which the preferences of each agent $i$ are captured by a {\em preference graph}, i.e., a directed acyclic graph $G_i$.
The vertex set of $G_i$ consists of all items that agent $i$ approves of, i.e., items the agent would like to receive.
Items not contained in the vertex set of the preference graph of an agent are regarded as disapproved of by the agent and we do not allocate such an item to that agent since it is irrelevant for her.
An arc $(a,b)$ in $G_i$ means that agent $i$ prefers item $a$ over item $b$. 
Assuming transitivity of the preferences, arcs $(a,b)$ and $(b,c)$ imply that the agent also prefers item $a$ over item $c$, regardless of whether the arc $(a,c)$ is contained in the graph or not. 
Observe that the graph of an agent hence induces a partial order over a subset of items. 

Every item can be assigned to at most one agent. 
For any such allocation of items to agents we consider the {\em dissatisfaction} of every agent.
An item causes dissatisfaction to an agent if it is not assigned to that agent and if the agent does not receive another more preferred item according to its preference graph.
The dissatisfaction of an agent is then determined by the total number of such items.
The goal of our optimization is the allocation of items to the agents in such a way that either the maximum dissatisfaction of an agent or the total dissatisfaction  of all agents (i.e., sum of the dissatisfaction over all the agents) is minimized.

From a graph theoretic perspective, an allocation for an agent $i$ is evaluated by the number of vertices in $G_i$ which are {\em dominated} by the allocated items, i.e., items which can be reached from an allocated vertex by a directed path in $G_i$. 
The total number of these dominated items together with the allocated items can be seen as the {\em satisfaction} level of agent $i$.
Obviously, satisfaction and dissatisfaction add up to the number of vertices in $G_i$.
However, minimizing the maximum dissatisfaction is in general not equivalent to maximizing the minimum satisfaction (see Section~\ref{sec:framework}), although for total dissatisfaction the equivalence holds.

\medskip
Our contribution is a computational complexity study aimed at narrowing down the complexity divide between classes of preference graphs where our allocation problems are still NP-complete and graph classes permitting polynomial-time solution algorithms.
It turns out that the latter is only possible for fairly simple types of graphs.
It should be noted that the graph classes analyzed in this paper are not at all esoteric special cases known only to experts in graph theory, but very natural restrictions corresponding to the typical human abilities of expressing preferences.

For instance, out-stars represent one item dominating a few less preferred items.
Out-trees can be seen as preference hierarchies,
directed matchings as pairwise comparisons and paths as a total order on a subset of items.
Throughout the paper we employ a wide variety of classical combinatorial optimization structures and approaches, such as matchings, independent sets, network flows, assignment, dynamic programming, and tree decompositions.
An overview of our results is given in Tables~\ref{fig:overview} and~\ref{fig:overview-tw}.
FPT is the class of problems parameterized by $t$ that are solvable in time $f(t) n^{\mathcal{O}(1)}$ for some computable function $f$ and instance size $n$ and XP the class of such problems solvable in time $n^{f(t)}$. A problem is paraNP-hard if it is NP-hard already for a constant value of the parameter. For an in-depth discussion of the parameterized complexity classes we refer the reader to~\cite{MR3380745}.

\begin{table}[htb]
\renewcommand{\arraystretch}{1.5}
\centering
\begin{tabular}{|l|c|c|}
\hline 
Setting: graph structure, $k$ agents & \MINMAXD & \MINSUMD \tabularnewline
\hline 
out-star, $k$ unbounded & NPc (Thm~\ref{th:hardmax}) & ? \tabularnewline
\hline 
out-tree, $k$ unbounded %, same set of items 
& NPc (Thm~\ref{th:hard-tree-B-all}) & NPc (Thm~\ref{th:hard-tree-B-all})\tabularnewline
\hline 
bipartite, $k=2$%, same set of items 
& NPc (Thm~\ref{thm:twokids-hard}) & NPc (Thm~\ref{thm:twokids-hard})\tabularnewline
\hline 
disjoint union of out-stars, $k=2$ &  ? & P (Thm~\ref{th:B-stars-2})\tabularnewline
\hline 
directed matching, $k$ unbounded & NPc (Thm~\ref{th:hard-A-matching}) & P (Thm~\ref{th:minSUM-matching})\tabularnewline
\hline 
directed matching, $k=2$ & P (Thm~\ref{th:A-matching-2}) & P (by Thm~\ref{th:minSUM-matching})  \tabularnewline
\hline 
path, $k$ unbounded & P (Thm~\ref{th:A-paths}) & P (Thm~\ref{th:A-paths})\tabularnewline
\hline 
disjoint union of paths, $k$ unbounded & NPc (Thm~\ref{th:max-paths}) & P (Cor~\ref{th:sum-paths})\tabularnewline
\hline
\end{tabular}
\caption{Overview of results presented in Sections~\ref{sec:hardness}--\ref{sec:special}. Here, ``NPc'' indicates that the respective decision problem is NP-complete, while ``P'' means that the problem can be solved in polynomial time. 
The open questions are indicated by a question mark.}
\label{fig:overview}
\end{table}

\begin{table}[htb]
\renewcommand{\arraystretch}{1.5}
\centering
\begin{tabular}{|l|c|c|}
\hline 
Parameterized by & \MINMAXD & \MINSUMD \tabularnewline
\hline
 $k$ & paraNPc (Thm~\ref{thm:twokids-hard}) & paraNPc (Thm~\ref{thm:twokids-hard}) \tabularnewline
\hline 
 $k+t$ & XP (Remark~\ref{rem:XP}) (FPT ?) & FPT (Thm~\ref{th:tw-minsum}) \tabularnewline
 \hline
 $d+k+t$ & FPT (Thm~\ref{th:tw-minmax}) & FPT (Thm~\ref{th:tw-minsum}) \tabularnewline
 \hline
 $\gamma$ & paraNPc (Thm~\ref{th:hard-A-matching})  & FPT (Thm~\ref{th:constjunc}) \tabularnewline
\hline 
\end{tabular}
\caption{Overview of parameterized complexity results (see, in particular, Section~\ref{sec:junction}). Here, ``paraNPc'', ``FPT'' and ``XP'' indicate membership in the respective complexity class, where ``paraNPc'' is shorthand for ``paraNP-complete'' (see Section~\ref{sec:junction} for short definitions). We denote by $k$ the number of agents, by $d$ the dissatisfaction threshold, by $\gamma$ the total number of vertices with in- or out-degree greater than one, 
% junctions in the instance, 
counted with multiplicities over all preference graphs, and by $t$ the treewidth of the undirected graph in which two distinct items are adjacent if and only if they form an arc in at least one of the preference graphs.
% $G=(V,A)$, consisting of vertex set $V$ and arc set $A$, the union of the arc sets $A_i$ of all graphs $G_i$ for $i \in K$. 
The open question is indicated by a question mark.}
\label{fig:overview-tw}
\end{table}

\medskip
\paragraph{Related work} In the literature on fair division of indivisible items, different approaches towards evaluating the quality of an allocation have been proposed. 
In the model of Herreiner and Puppe~\cite{puppe}, the agents need to rank all possible subsets of items. Avoiding this task, in the standard framework the agents rank the single items instead, and some approaches additionally require that the agents partition the set of items into approved and disapproved ones (\cite{Brams2009,Dong2021}). 
On the one hand, certain fairness notions such as envy-freeness (no agent envies another agent for her received set of items), equitability (the disparity between the happiest and unhappiest agent is minimized), or proportionality (each agent receives a set of items she values at least $\frac{1}{k}$-th of the value of the whole set of items, where $k$ denotes the number of agents)  have been studied; we refer to the surveys  \cite{bouveret-survey} and \cite{thomson} for details. 
On the other hand, the satisfaction level (or individual welfare) of each agent has been taken into account in order to determine the social welfare induced  (see, e.g., Aziz et al.~\cite{aziz}, Brams et al.~\cite{brams-twoperson}, or Baumeister et al.~\cite{baumeister2017}). 
This is typically done by assuming cardinal preference information, introducing utility functions, or using scores from voting schemes to evaluate ordinal preferences. 
Often, the social welfare---most prominently, utilitarian (i.e., sum of individual welfares), egalitarian (minimum individual welfare), and Nash social welfare (product of individual welfares)---has been analyzed from a computational viewpoint (Baumeister et al.~\cite{baumeister2017}, 
Bansal and Sviridenko~\cite{santaclaus}, 
Chiarelli et al.~\cite{fairIWOCA},
Darmann and Schauer~\cite{darmann2015},  Garg and McGlaughlin~\cite{garg2019}, Roos and Rothe~\cite{roosarothe}).
We add to that literature and point out that in our model the agents do not need to give a full ranking of the items but face the simpler task of stating  partial orderings of the items only (represented by means of preference graphs). Our model, however, has a certain vicinity to the settings of \cite{Brams2009} and \cite{Dong2021} as the agents' preference graphs do not need to contain all of the vertices and the non-contained vertices are regarded as disapproved of. Observe that our objectives of minimizing total and maximum dissatisfaction among the agents can be understood as maximizing utilitarian and egalitarian social welfare respectively. 

\bigskip
The paper is structured as follows.
After presenting the formal framework of the paper including problem definitions in Section~\ref{sec:framework}, we show in Section~\ref{sec:hardness} that the corresponding decision problem {\MINMAX} resp.\ {\MINSUM} is NP-complete when the preference graphs are out-stars resp.\ out-trees.
Restricting the allocation to only two agents, both problem versions remain NP-complete, but become polynomial-time solvable for disjoint unions of out-stars when dealing with {\MINSUM}.
Section~\ref{sec:matching} gives an interesting dichotomy for directed matchings as preference graphs.
Minimizing the total dissatisfaction can be done in polynomial time, whereas minimizing the maximum dissatisfaction is shown to be NP-complete.
Only for the restriction to two agents, a positive result can be derived for that case.
In Section~\ref{sec:special} we consider the special case of paths. 
Here, both problem variants remain polynomial-time solvable, but surprisingly, {\MINMAX} already becomes NP-complete as soon as every preference graph consists of at most two paths.
More general results are presented in Section~\ref{sec:junction}.
In particular, we can show fixed-parameter tractability for {\MINSUM} in the total number of junction vertices in all preference graphs by utilizing a max-profit network flow model.
Furthermore, a fixed-parameter tractability result can be shown with respect to the treewidth of the union of preference graphs and the number of agents.
A preliminary version of this work containing some of the results and omitting most of the proofs appeared as~\cite{adt2021}.

\section{Formal framework}
\label{sec:framework}

In this section we describe notation and definitions.
An \textit{undirected graph} is a pair $G=(V,E)$ with vertex set $V$ and edge set $E\subseteq \{\{u,v\}\mid u,v \in V, u\neq v\}$.
For a graph $G = (V,E)$ we write $V(G)$ for $V$ and $E(G)$ for $E$.
A \textit{directed graph} is a pair  $G=(V,A)$ with vertex set $V$ and set of arcs $A\subseteq V\times V$.
For a directed graph $G = (V,A)$ we write $V(G)$ for $V$ and $A(G)$ for $A$.
For brevity, we will often say \textit{graph} when referring to a directed graph.

Consider a directed graph $G=(V,A)$. For $a=(u,v)\in A$, vertex $u$ is called tail of $a$ and vertex $v$ is called head of $a$. 
The \textit{in-degree} of a vertex $u$ is the number of arcs in $A$ for which $u$ is head and the \textit{out-degree} of $u$ is the number of arcs in $A$ for which $u$ is tail.
The \textit{degree} of a vertex $u$ is the number of arcs in $A$ for which $u$ is either head or tail. 
A sequence $p=(v_0,v_1,v_2,\ldots,v_\ell)$ with $\ell \ge 0$
and $(v_i,v_{i+1})\in A$ for each $i\in \{0,\ldots,\ell-1\}$ is called a \textit{walk} of length $\ell$ from $v_0$ to $v_\ell$; it is a \textit{path} (of length $\ell$ from $v_0$ to $v_\ell$) if all its vertices are pairwise distinct.
A walk from  $v_0$ to $v_\ell$ is \textit{closed} if $v_0 = v_\ell$.
A \textit{cycle} is a closed walk of positive length in which all vertices are pairwise distinct, except that $v_0 = v_\ell$.
A \textit{directed acyclic graph} is a directed graph with no cycle.
An  \textit{out-tree} is a directed acyclic graph $G=(V,A)$ with a dedicated vertex $r$ (called root) such that for each vertex $v\in V\setminus\{r\}$ there is exactly one path from $r$ to $v$. An \textit{out-star} is an out-tree in which each such path is of length $1$. 
A \emph{matching} in an undirected graph $G$ is a set of pairwise disjoint edges.
A \textit{directed matching} is a directed acyclic graph such that each vertex has degree exactly one (i.e., the edges of the underlying undirected graph $G$ form a matching in $G$).

 A binary relation ${\succ}\subseteq V\times V$ is a \textit{strict partial order} over $V$ if it is  
 asymmetric (for all $u,v\in V$, if $u \succ v$ then $v \succ u$ does not hold) and transitive (for all $u,v,w\in V$, $u \succ v$ and $v \succ w$ imply $u\succ w$). 
 Observe that a directed acyclic graph $G=(V,A)$ induces a strict partial order $\succ$ on $V$ by setting $u \succ v$ for each pair $(u,v) \in V\times V$ such that $u\neq v$ and there is a path from $u$ to $v$.
\newMM{In particular, $u \succ v$ for every arc $(u,v) \in A$.}
%  (which includes every arc $(u,v) \in A$).
 
 In what follows, we will consider a vertex set $V$ and a set of agents $K$ along with directed acyclic graphs $G_i=(V_i,A_i)$ for all $i \in K$, where $V_i \subseteq V$ \newAD{represents the items desired by agent $i$.}
 Let ${\it pred}_i(v) \subseteq V_i$ denote the set of predecessors of $v$ in graph $G_i$, i.e., the set of all vertices $u\neq v$ such that there is a path from $u$ to $v$ in $G_i$. 
 Observe that  ${\it pred}_i(v)$ corresponds to the set of items which agent $i$ \textit{prefers} over $v$ under relation~$\succ$.
  For $u,v \in V_i$ we say that item $u$ is \textit{dominated} by  item $v$ if $u=v$ or $v\in {\it pred}_i(u)$.
In addition, let ${\it succ}_i(v) \subseteq V_i$ denote the set of successors of $v$ in graph $G_i$, i.e., the set of all vertices $u\neq v$ such that there is a path from $v$ to $u$ in $G_i$. Hence,  ${\it succ}_i(v)$ denotes the set of all items to which agent $i$ prefers item $v$.

\smallskip
An \textit{allocation} $\pi$ is a function $K \rightarrow  2^V$ that assigns to the agents pairwise disjoint sets of items, i.e., for $i,j \in K$, $i\not=j$, we have $\pi(i)\cap\pi(j)=\emptyset$.
To measure the attractiveness of an allocation we will count the number of items, which an agent {\bf does not receive} and for which she receives no other more preferred item. 

Formally, for an allocation $\pi$  the \textit{dissatisfaction} $\diss_{\pi}(i)$ of agent $i$ is defined as number of items \newAD{in $G_i$} not dominated by any item in $\pi(i)$.  
A \textit{dissatisfaction profile} is a $|K|$-tuple $(d_i\mid i\in K)$ with $d_i \in \nat_0$ for all $i\in K$ such that there is an allocation $\pi$ with  $\diss_{\pi}(i)=d_i$ for each $i\in K$.

Note that the underlying undirected graphs of the directed graphs $G_i$ are not necessarily connected and may also contain isolated vertices.
According to the definition of $\diss_{\pi}(i)$, every isolated vertex in $G_i$ contributes one unit to $\diss_{\pi}(i)$, if it is not allocated to agent $i$.
Vertices in $V\setminus V_i$ are irrelevant for agent $i$; they do not have any influence on the dissatisfaction function $\diss_{\pi}(i)$.

In this paper, we focus on the following two problems aiming to minimize the maximum and total dissatisfaction among the agents.

\bigskip

\noindent{\MINMAX}: 

\noindent \textbf{Input:} A set $K$ of agents, a set $V$ of items, a  directed  acyclic graph $G_{i} =(V_i,A_i)$ for each $i\in K$ with $V_i \subseteq V$, and an integer $d$.

\noindent \textbf{Question:} Is there an allocation $\pi$ of items to agents such that the dissatisfaction $\diss_{\pi}(i)$ is at most $d$ for each agent $i\in K$? 

\medskip
\noindent{\MINSUM}: 

\noindent \textbf{Input:} A set $K$ of agents, a set $V$ of items, a  directed  acyclic graph $G_{i} =(V_i,A_i)$ for each $i\in K$ with $V_i \subseteq V$, and an integer $d$.

\noindent \textbf{Question:} Is there an allocation $\pi$ of items to agents such that the total dissatisfaction $\sum_{i\in K} \diss_{\pi}(i)$ is at most $d$?

\medskip

The graphs $G_i$ in the above problem definitions are called \textit{preference graphs}. Throughout the paper, we denote the number of agents by $k=|K|$ and the number of items by $n=|V|$.

While in this work the focus is laid on the minimization of (maximum or total) dissatisfaction, some remarks on the associated dual problem of maximizing satisfaction are in order. 
In this context, satisfaction $s_\pi (i)$ of agent $i$ with respect to allocation $\pi$ is measured by means of the number of items in $V_i$  that are dominated by some item in $\pi(i)$.  \newAD{Observe that the items in $V\setminus V_i$ are irrelevant for the satisfaction of agent $i$.}

Observe that minimizing the total dissatisfaction is equivalent to maximizing the total satisfaction of all agents. 
In order to verify this, it is sufficient to note that $\diss_\pi (i)=|V_i|-s_\pi (i)$, which implies that $\sum_{i\in K} \diss_{\pi}(i)=\sum_{i\in K} |V_i| - \sum_{i\in K} s_{\pi}(i)$ is minimized if and only if $\sum_{i\in K} s_{\pi}(i)$ is maximized. 

On the other hand, we point out that minimizing the maximum dissatisfaction in general does not correspond to maximizing the minimum satisfaction among the agents, as the following example shows. 

\medskip
\noindent{\bf Example.}
Let $V=\{a,b,c\}$ and $K=\{1,2,3\}$, with each graph $G_i$ consisting of isolated vertices only. Graph $G_1$ consists of all three vertices $a,b,c$, graphs $G_2$  and $G_3$ are made up of $b$ and $c$ respectively.  An allocation with dissatisfaction of at most $1$ per agent is given only by allocations that give at least two items to agent~$1$.
On the other hand, the allocation $\pi$ with $\pi(1)=\{a\}$, $\pi(2)=\{b\}$, $\pi(3)=\{c\}$ is the only allocation that yields a satisfaction of $1$ for each agent, resulting in a dissatisfaction of $2$ for agent $1$. 
%\end{example}

\medskip
\begin{sloppypar}
In what follows, we provide NP-completeness results on the one hand and positive results, i.e., polynomial-time solvable cases, on the other. 
Concerning these positive results we point out that we are not only able to answer the corresponding decision question but also to solve the associated optimization problem, i.e., we can find an allocation that minimizes total resp.\ maximum dissatisfaction in polynomial time.
With respect to the NP-completeness results, note that  from the definition of {\MINMAX} and {\MINSUM} it follows that any NP-hardness result implies NP-hardness in the strong sense. 
Several of the hardness results presented in this paper reduce from the following NP-complete variant (3X3C) of \textsc{Exact Cover by $3$-Sets} (see Gonzalez~\cite{gonzalez85}). 
\end{sloppypar}

\medskip

\noindent \textsc{Exact Cover by $3$-Sets (3X3C)}:

\noindent \textbf{Input:} A set $X$ with $|X| = 3q$, and 
a collection $C=\{C_1,\ldots,C_p\}$ of $3$-element subsets of $X$ such that each element of $X$ appears in exactly $3$ sets.

\noindent \textbf{Question:} Does $C$ contain an \textit{exact cover} of $X$, that is, a subcollection $C'\subseteq C$ such that every element of $X$ occurs in exactly one member of $C'$? %\medskip
%Note that in each instance of 3X3C we have $p=3q$ and therefore $p$ is a multiple of $3$.

\section{Two agents, trees, and stars} 
\label{sec:hardness}

We begin our computational complexity study in this romantically named section with out-trees as preference graphs. 
Out-trees constitute a very natural model of partial preferences. 
It turns out that already in this rather simple case both {\MINMAX} and {\MINSUM} are NP-complete; for the former, NP-completeness even holds when all graphs are restricted to out-stars.

\begin{theorem}\label{th:hardmax}
{\MINMAX} is NP-complete, even if each graph $G_i$ is an out-star.
\end{theorem}

\begin{proof} Given an instance $\mathcal{J}$ of 3X3C with a set $X$ of elements, $|X|=3q$, and a collection $C=\{C_1,\ldots,C_p\}$ of $3$-element subsets of $X$, we construct an instance $\mathcal{I}$ of {\MINMAX} as follows. Recall that we have $p=3q$. Set $\ell=\frac{2}{3} p +1$, and let the set of items $V=X\cup\{1,\ldots,p\}\cup\{h_1,\ldots,h_{\ell+1}\}$. 
The set of agents $K$ is made up of the agents $D_1,\ldots,D_{\ell+1}$, agent $A$, and agents $C_1,\ldots,C_p$ (the latter agents are identified with the sets of the same label 
in instance $\mathcal{J}$). The graphs $G_i$ are out-stars displayed in Figure~\ref{mcs:figure4}; the graph of agent $C_j$ has root vertex $j$, and contains the edges $(j,h_r)$, for $r\in \{1,\ldots,{\ell-1}\}$ and $(j,a)$ for $a\in C_j$ (for the figure, we assume set $C_1=\{x,y,z\}$ and  $C_2=\{x,u,v\}$).
We ask if there is an allocation with dissatisfaction at most $\ell$ per agent.

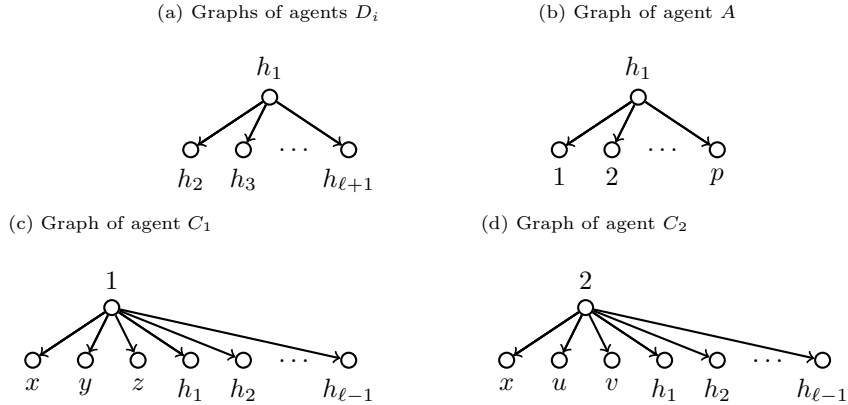
\begin{figure}[H]
	\centering
	\begin{tikzpicture}[vertex/.style={inner sep=2pt,draw,circle},scale=0.7]

	\begin{scope}[xshift=-2cm]
	
	\node[vertex, label=above:$ h_1 $] (1) at (5,5) {};
	\node[vertex, label=below:$ h_2 $] (2) at (3.5,4) {};
	\node[vertex, label=below:$ h_3 $] (3) at (4.5,4) {};
	\node[] (4) at (5.5,4) {$\dots$};
	\node[vertex, label=below:$ h_{\ell+1} $] (5) at (6.5,4) {};
	\draw[] (3) --  (1) -- (2);
	\draw[] (5) --  (1) -- (2);
	\path[->] (1) edge (2);
	\path[->] (1) edge (3);
	\path[->] (1) edge (5);
	%\node [below of=s2c] {\parbox{0.3\linewidth}{	\subcaption{agents $D_1,\ldots,D_\ell$}}
	   % \node [below of=1] {\parbox{0.3\linewidth}{\subcaption{agents }}\label{subfig:a}}};
	       \node [above of=1] {\parbox{0.3\linewidth}{\subcaption{Graphs of agents $D_i$}\label{subfig:a}}};
	\end{scope}
	
	\begin{scope}[xshift=5cm]
	\node[vertex, label=above:$ h_1 $] (1) at (5,5) {};
	\node[vertex, label=below:$ 1 $] (2) at (3.5,4) {};
	\node[vertex, label=below:$ 2 $] (3) at (4.5,4) {};
	\node[] (4) at (5.5,4) {$\dots$};
	\node[vertex, label=below:$ p $] (5) at (6.5,4) {};
	\draw[] (3) --  (1) -- (2);
	\draw[] (5) --  (1) -- (2);
	\path[->] (1) edge (2);
	\path[->] (1) edge (3);
	\path[->] (1) edge (5);
	 \node [above of=1] {\parbox{0.3\linewidth}{\subcaption{Graph of agent $A$}\label{subfig:b}}};
	\end{scope}

		\begin{scope}[xshift=-5cm,yshift=-4cm]
	\node[vertex, label=above:$ 1 $] (1) at (5,5) {};
	\node[vertex, label=below:$ x $] (2) at (3.5,4) {};
	\node[vertex, label=below:$ y $] (3) at (4.5,4) {};
	\node[vertex, label=below:$ z $] (4) at (5.5,4) {};
	\node[vertex, label=below:$ h_1 $] (5) at (6.5,4) {};
	\node[vertex, label=below:$ h_2 $] (6) at (7.5,4) {};

	\node[] (7) at (8.5,4) {$\dots$};
	\node[vertex, label=below:$ h_{\ell-1} $] (9) at (9.5,4) {};
	\draw[] (3) --  (1) -- (2);
	\draw[] (5) --  (1) -- (2);
	\path[->] (1) edge (2);
	\path[->] (1) edge (3);
	\path[->] (1) edge (4);
		\path[->] (1) edge (5);
		\path[->] (1) edge (6);
		\path[->] (1) edge (9);

	%\node [below of=s2c] {\parbox{0.3\linewidth}{\subcaption{agents $D_1,\ldots,D_\ell$}}
	   % \node [below of=1] {\parbox{0.3\linewidth}{\subcaption{agents }}\label{subfig:a}}};
	       \node [above of=1] {\parbox{0.3\linewidth}{\subcaption{Graph of agent $C_1$}\label{subfig:c}}};
	\end{scope}
% 	\begin{scope}[xshift=-5cm,yshift=-7cm]
% 	\node[vertex, label=left:$ d_1 $] (d1) at (5,10) {};
% 	\node[vertex, label=left:$ d_2 $] (d2) at (5,9) {};
% 	\node[] (d3) at (5,8) {$ \vdots $};
% 
% 	\node[vertex, label=left:$ d_{l-2} $] (d4) at (5,7) {};
% 	\node[vertex, label=left:$ d_{l-1} $] (d5) at (5,6) {};
% 	\node[vertex, label=left:$ 1 $] (1) at (5,5) {};
% 	\node[vertex, label=below:$ x $] (2) at (4,4) {};
% 	\node[vertex, label=below:$ y $] (3) at (5,4) {};
% 	\node[vertex, label=below:$ z $] (4) at (6,4) {};
% 	\draw[] (d1) --  (d2) -- (d3)--(d4)--(d5)--(1);
% 	\draw[] (3) --  (1) -- (2);
% 	\draw[] (3) --  (1) -- (2);
% 	\draw[] (4) --  (1) -- (2);
% 	\end{scope}
	
% 	\begin{scope}[yshift=-7cm]
% 	\node[vertex, label=above:$ 2 $] (1) at (5,5) {};
% 	\node[vertex, label=below:$ x $] (2) at (4,4) {};
% 	\node[vertex, label=below:$ z $] (3) at (5,4) {};
% 	\node[vertex, label=below:$ u $] (4) at (6,4) {};
% 	\draw[] (3) --  (1) -- (2);
% 	\draw[] (4) --  (1) -- (2);
% 	\end{scope}
	\begin{scope}[xshift=4cm,yshift=-4cm]
	\node[vertex, label=above:$ 2 $] (1) at (5,5) {};
	\node[vertex, label=below:$ x $] (2) at (3.5,4) {};
	\node[vertex, label=below:$ u $] (3) at (4.5,4) {};
	\node[vertex, label=below:$ v $] (4) at (5.5,4) {};
	\node[vertex, label=below:$ h_1 $] (5) at (6.5,4) {};
	\node[vertex, label=below:$ h_2 $] (6) at (7.5,4) {};

	\node[] (7) at (8.5,4) {$\dots$};
	\node[vertex, label=below:$ h_{\ell-1} $] (9) at (9.5,4) {};
	\draw[] (3) --  (1) -- (2);
	\draw[] (5) --  (1) -- (2);
	\path[->] (1) edge (2);
	\path[->] (1) edge (3);
	\path[->] (1) edge (4);
		\path[->] (1) edge (5);
		\path[->] (1) edge (6);
		\path[->] (1) edge (9);

	%\node [below of=s2c] {\parbox{0.3\linewidth}{	\subcaption{agents $D_1,\ldots,D_\ell$}}, displayed in Figure~\ref{mcs:figure4}
	   % \node [below of=1] {\parbox{0.3\linewidth}{\subcaption{agents }}\label{subfig:a}}};
	    %   \node [above of=1] {\parbox{0.3\linewidth}{\subcaption{agents $C_1,\ldots,C_{p}$}\label{subfig:a}}};
	    	       \node [above of=1] {\parbox{0.3\linewidth}{\subcaption{Graph of agent $C_2$}\label{subfig:d}}};

	\end{scope}

% 	\begin{scope}[xshift=7cm, yshift=-7cm]
% 	\node[] (1) at (1.5,4.5) {$\dots$};
% 	\node[vertex, label=above:$ p $] (1) at (5,5) {};
% 	\node[vertex, label=below:$ u $] (2) at (4,4) {};
% 	\node[vertex, label=below:$ v $] (3) at (5,4) {};
% 	\node[vertex, label=below:$ w $] (4) at (6,4) {};
% 	\draw[] (3) --  (1) -- (2);
% 	\draw[] (4) --  (1) -- (2);
% 	\end{scope}
	\end{tikzpicture}\caption{Graphs of agents in the proof of \newwSL{Theorem~\ref{th:hardmax}}.}\label{mcs:figure4}
\end{figure}
%\vspace{-0.2cm}
We show that $\mathcal{J}$ is a yes-instance of 3X3C if and only if $\mathcal{I}$ is a yes-instance of {\MINMAX}. 
\newMM{Assume first} 
% ``$\Leftarrow$'': Suppose 
that there is an allocation with dissatisfaction at most $\ell$ per agent. Agents $D_1,\ldots, D_{\ell+1}$ all have the same graph (displayed in Figure~\ref{mcs:figure4}). To respect the dissatisfaction bound $\ell$, each of the agents $D_1,\ldots, D_{\ell+1}$ has to get exactly one of the items $h_1,\ldots,h_{\ell+1}$. Hence, also item $h_1$ is already allocated. Therefore, agent $A$ has to get at 
least $\frac{1}{3} p $ items of $\{1,\ldots,p\}$ in order to respect the bound $\ell$. Thus, for the agents $C_1,\ldots,C_p$ there are
only at most $\frac{2}{3} p$ items of $\{1,\ldots,p\}$ left. As a consequence, in order to respect bound $\ell$, 
at least $\frac{1}{3} p$ of these agents need to get all three items which 
make up the respective set in instance  $\mathcal{J}$ (e.g., agent $C_1$ gets all three items $x,y,z$). 
As each present can be given to one agent and the set $X$ contains exactly $p$ elements, this means that the collection of sets $C_j$ such that the agent of the same label gets all three items of $C_j$ forms an exact cover of $X$. 

\newMM{For the converse direction, let} 
% ``$\Rightarrow$'': Let
$B$ be a collection of sets of $C$ that forms an exact cover of $X$. Observe that $B$ contains exactly $\frac{p}{3}$ sets. Give $h_i$ to agent $D_i$ for each $i$. For collection $B$,
(i) give, for each $C_j \in B$ its three elements (items) to agent $C_j$ (there are $\frac{p}{3}$ of such sets/agents), and 
(ii) give all items $j$ such that $C_j \in B$ to agent $A$ (these are $\frac{p}{3}$ such items), (iii) give item $j$ such that $C_j \notin B$ to agent $C_j$. 
Then, for each agent the dissatisfaction bound $\ell$ is respected. %\hfill $\square$
\qed
\end{proof}

\medskip
In the proof of Theorem~\ref{th:hardmax} we used out-stars consisting only of specially selected items.
However, using further reductions from 3X3C we can also show that both {\MINSUM} and {\MINMAX}
are computationally hard, even when  all graphs $G_i$ are very simple types of out-trees each containing {\bf all vertices} of $V$.

\begin{theorem}\label{th:hard-tree-B-all}
{\MINSUM} and {\MINMAX} are NP-complete, even if each graph $G_i$ is an out-tree containing all vertices of $V$. 
\end{theorem}

\begin{proof}
We give a proof for \MINSUM, the proof for \MINMAX\ proceeds analogously.

Again we provide a reduction from 3X3C. Given an instance $\mathcal{J}$ of 3X3C with a set $X$ of elements and a collection $C=\{C_{1},\ldots,C_{p}\}$ of $3$-element subsets of $X$, let $D=2p^{2}+\frac{p}{3}+1$.  We construct an instance $\mathcal{I}$ of the  {\MINSUM}
by introducing a set of items $V=X\cup\{h_{i}\mid1\leq i\leq 3D-2\}\cup\{1,\ldots,p\}\cup\{e\}$
and a set $K$ made up of the following $3D-2+p+1= 6p^{2}+2p+2$ agents: agents $A_{i}$ for $1\leq i \leq 3D-2$, agent $B$, and agents $C_j$---representing set $C_j$ in instance $\mathcal{J}$ of 3X3C---for $1\leq j \leq p$. \newAD{Except for agent $A_{3D-2}$,} the agents' graphs are  displayed in Figure~\ref{fig:tree}, where $\text{path}_a$ and $\text{path}_b$ denote arbitrary but fixed paths along all vertices of $X\cup\{1,\ldots,p\}$ and $X$ respectively, and, for $j\in \{1,\ldots,p\}$,  $\text{path}_j$ denotes an arbitrary but fixed path along all vertices of $(X\setminus C_j)\cup\{1,\ldots,p\}\setminus\{j\}$. \newAD{The graph of agent $A_{3D-2}$ is an out-tree with root $h_{3D-2}$ made up of the arcs $(h_{3D-2},h_i)$ for $1\leq i \leq 3D-3$, and of the path $(h_{3D-3},path_a,e)$. 
Observe that in the graph of agent $B$, for each $1\leq i\leq p$, the unique path starting from  $i$ is of length $3p-1$; vertex $e$ functions as dummy, with the path starting from $e$ containing the remaining vertices of $V$.
}
%$G{A_{i}}$
%has $d_{i}$ as root vertex and its edge set is made up of a directed
%edge from $d_{i}$ to $d_{j}$ for each $j\not=i$. 

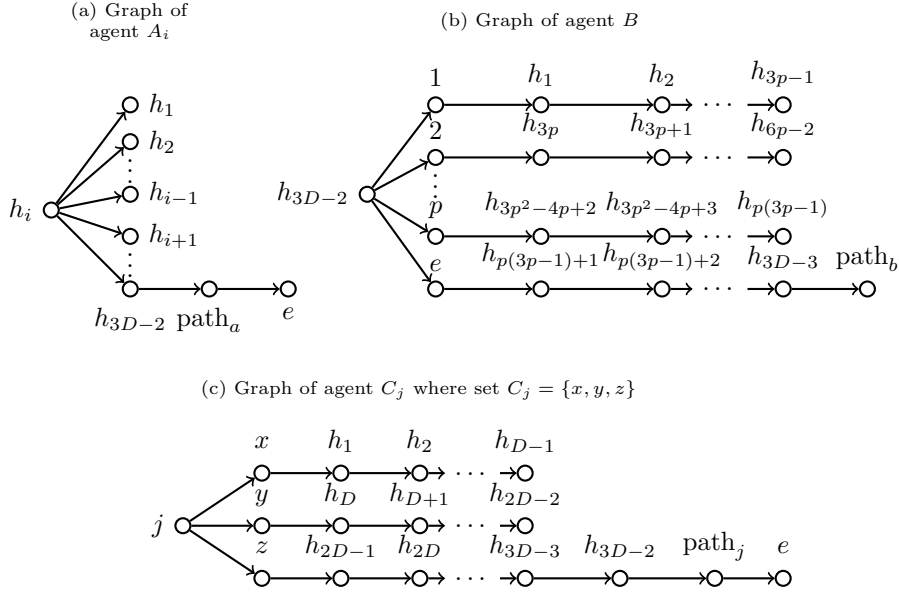
\begin{figure}[htbp]
	\centering
	\begin{tikzpicture}[vertex/.style={inner sep=2pt,draw,circle},scale=0.7]

	\begin{scope}[xshift=-12cm]
	\node[vertex, label=left:$ h_i $] (1) at (4,2.5) {};
	\node[vertex, label=right:$ h_1 $] (2) at (5.5,4.5) {};
	\node[vertex, label=right:$ h_2 $] (3) at (5.5,3.8) {};
	\node[] (4) at (5.5,3.4) {$\vdots$};
	\node[vertex, label=right:$ h_{i-1} $] (5) at (5.5,2.8) {};
	\node[vertex, label=right:$ h_{i+1} $] (6) at (5.5,2) {};
	\node[] (7) at (5.5,1.6) {$\vdots$};
	\node[vertex, label=below:$ h_{3D-2} $] (9) at (5.5,1) {};
	\node[vertex, label=below:$\text{path}_a$] (10) at (7,1) {};
	\node[vertex, label=below:$e$] (11) at (8.5,1) {};
		\path[->] (1) edge (2);
		\path[->] (1) edge (3);
		\path[->] (1) edge (5);
		\path[->] (1) edge (6);
		\path[->] (1) edge (9);
		\path[->] (9) edge (10);
		\path[->] (10) edge (11);
	\node [above of=2] {\parbox{0.15\linewidth}{\subcaption{Graph of agent $A_i$}\label{subfig:a}}};
	\end{scope}
	
	\begin{scope}[xshift=-6cm]
	\node[vertex, label=left:$ h_{3D-2} $] (1) at (4,2.8) {};
	\node[vertex, label=above:$ 1 $] (2) at (5.3,4.5) {};
	\node[vertex, label=above:$ 2 $] (3) at (5.3,3.5) {};
	\node[] (4) at (5.3,3.1) {$\vdots$};
	\node[vertex, label=above:$ p $] (5) at (5.3,2) {};
	\node[vertex, label=above:$ e $] (18) at (5.3,1) {};
		\path[->] (1) edge (2);
		\path[->] (1) edge (3);
		\path[->] (1) edge (5);	
		\path[->] (1) edge (18);	
	\node[vertex, label=above:$ h_1 $] (6) at (7.3,4.5) {};
	\node[vertex, label=above:$ h_2 $] (7) at (9.6,4.5) {};
	\node[] (8) at (10.7,4.5) {$\dots$};
	\node[vertex, label=above:$ h_{3p-1} $] (9) at (11.9,4.5) {};
		\path[->] (2) edge (6);
    	\path[->] (6) edge (7);
		\path[->] (7) edge (8);
		\path[->] (8) edge (9);    
	\node[vertex, label=above:$ h_{3p} $] (10) at (7.3,3.5) {};
	\node[vertex, label=above:$ h_{3p+1} $] (11) at (9.6,3.5) {};
	\node[] (12) at (10.7,3.5) {$\dots$};
	\node[vertex, label=above:$ h_{6p-2} $] (13) at (11.9,3.5) {};
    	\path[->] (3) edge (10);
    	\path[->] (10) edge (11);
		\path[->] (11) edge (12);
		\path[->] (12) edge (13);
	\node[vertex, label=above:$ h_{3p^2-4p+2} $] (14) at (7.3,2) {};
	\node[vertex, label=above:$ h_{3p^2-4p+3} $] (15) at (9.6,2) {};
	\node[] (16) at (10.7,2) {$\dots$};
	\node[vertex, label=above:$ h_{p(3p-1)} $] (17) at (11.9,2) {};
    	\path[->] (5) edge (14);
    	\path[->] (14) edge (15);
		\path[->] (15) edge (16);
		\path[->] (16) edge (17);	
	\node[vertex, label=above:$ h_{p(3p-1)+1} $] (19) at (7.3,1) {};
	\node[vertex, label=above:$ h_{p(3p-1)+2} $] (20) at (9.6,1) {};
	\node[] (21) at (10.7,1) {$\dots$};
	\node[vertex, label=above:$ h_{3D-3} $] (22) at (11.9,1) {};
	\node[vertex, label=above:$ \text{path}_b $] (23) at (13.5,1) {};			
		\path[->] (18) edge (19);
		\path[->] (19) edge (20);
		\path[->] (20) edge (21);
		\path[->] (21) edge (22);
		\path[->] (22) edge (23);

	\node [above of=6] {\parbox{0.3\linewidth}{\subcaption{Graph of agent $B$}\label{subfig:b}}};
	\end{scope}
	
	\begin{scope}[xshift=-9.5cm,yshift=-6cm]
	\node[vertex, label=left:$\mathstrut j $] (1) at (4,2.5) {};
	\node[vertex, label=above:$\mathstrut x $] (2) at (5.5,3.5) {};
	\node[vertex, label=above:$\mathstrut y $] (3) at (5.5,2.5) {};
	\node[vertex, label=above:$\mathstrut z $] (4) at (5.5,1.5) {};
		\path[->] (1) edge (2);
		\path[->] (1) edge (3);
		\path[->] (1) edge (4);
	\node[vertex, label=above:$\mathstrut h_1 $] (6) at (7,3.5) {};
	\node[vertex, label=above:$\mathstrut h_2 $] (7) at (8.5,3.5) {};
	\node[] (8) at (9.5,3.5) {$\dots$};
	\node[vertex, label=above:$\mathstrut h_{D-1} $] (9) at (10.5,3.5) {};
    	\path[->] (2) edge (6);
    	\path[->] (6) edge (7);
		\path[->] (7) edge (8);
		\path[->] (8) edge (9);
	\node[vertex, label=above:$\mathstrut h_D $] (10) at (7,2.5) {};
	\node[vertex, label=above:$\mathstrut h_{D+1}$] (11) at (8.5,2.5) {};
	\node[] (12) at (9.5,2.5) {$\dots$};
	\node[vertex, label=above:$\mathstrut h_{2D-2} $] (13) at (10.5,2.5) {};
    	\path[->] (3) edge (10);
    	\path[->] (10) edge (11);
		\path[->] (11) edge (12);
		\path[->] (12) edge (13);
	\node[vertex, label=above:$\mathstrut h_{2D-1} $] (14) at (7,1.5) {};
	\node[vertex, label=above:$\mathstrut h_{2D} $] (15) at (8.5,1.5) {};
	\node[] (16) at (9.5,1.5) {$\dots$};
	\node[vertex, label=above:$\mathstrut h_{3D-3} $] (17) at (10.5,1.5) {};
	\node[vertex, label=above:$\mathstrut h_{3D-2} $] (18) at (12.3,1.5) {};
	\node[vertex, label=above:$ \text{path}_j $] (19) at (14.1,1.5) {};
	\node[vertex, label=above:$\mathstrut e$] (20) at (15.4,1.5) {};
    	\path[->] (4) edge (14);
    	\path[->] (14) edge (15);
		\path[->] (15) edge (16);
		\path[->] (16) edge (17);
		\path[->] (17) edge (18);
		\path[->] (18) edge (19);
		\path[->] (19) edge (20);
	\node [above of=7] {\parbox{0.7\linewidth}{\subcaption{Graph of agent $C_j$ where set $C_j = \{x,y,z\}$}\label{subfig:c}}};
	\end{scope}
	
	\end{tikzpicture}\caption{Graphs of agents in the proof of Theorem~\ref{th:hard-tree-B-all}.}\label{fig:tree}
\end{figure}

We prove that $C$ contains an exact cover of $X$ if and only if $\mathcal{I}$ admits an allocation $\pi$ with $\sum_{i \in K}\delta_{\pi}(i)\leq D$. 

Assume that $Y$ is a collection of sets of $C$ that forms an exact cover of $X$. Then consider the following allocation $\pi$: 
\begin{itemize}
\item allocate item $h_{i}$ to agent $A_{i}$, for each $i$; hence, we have
$\delta_{\pi}(A_{i})=0$ for each $i$;
\item give item $e$ and all items $j$ with $C_{j}\in Y$ to agent $B$. Observe that due
to $|Y|=\frac{p}{3}$ agent $B$ receives exactly $\frac{p}{3}$ items,
and $\delta_{\pi}(B)=1+\frac{2}{3}p3p=2p^{2}+1$;
\item give item $j$ with $C_{j}\notin Y$ to agent $C_{j}$, which yields
$\delta_{\pi}(C_{j})=0$;
\item for each $C_{j}\in Y$ give the three items corresponding to the elements
of $C_{j}$ to agent $C_{j}$, which yields $\delta_{\pi}(C_{j})=1$;
\end{itemize}
In total, we hence get $\sum_{i \in K}\delta_{\pi}(i)=0+2p^{2}+1+\frac{p}{3}$.

On the other hand, assume that $\mathcal{I}$ admits an allocation
$\pi$ with $\sum_{i \in K}\delta_{\pi}(i)\leq D$. This requires that each
of the agents $A_{i}$ must get at least one of $\{h_{j}\mid1\leq j\leq 3D-2\}$.
Since all the elements of $\{h_{j}\mid1\leq j\leq3D-2\}$ are allocated to the agents $A_i$, agent $B$ hence has to receive at least $\frac{1}{3} p$
items of $\{1,\ldots,p\}$, because otherwise we would have
$\delta_{\pi}(B)\geq(\frac{2}{3}p+1)3p=2p^{2}+3p>D$, in contradiction
with our assumption. Let $C_{j}=\{x,y,z\}$, and consider agent
$C_{j}$. If $C_{j}$ does not receive item $j$ she must get all
three items $x,y,z$, because otherwise $\delta_{\pi}(C_{j})\geq D+1 >D$
would hold. Now, observe that at least $\frac{1}{3} p$ elements of  $\{1,\ldots,p\}$
are already allocated to $B$, and therefore at most $\frac{2}{3}$ of all the agents $C_{j}$, $1\leq j\leq p$, can receive item $j$. Thus, at least
$\frac{1}{3}$ of all agents $C_{j}$ need to get the three items that
make up the respective set $C_{j}$ in instance $\mathcal{I}$ of
3X3C. Since there are $|X|=p$ such items, however, exactly $\frac{1}{3}$ of
all agents $C_{j}$ get the three items that make up the respective
set $C_{j}$ in instance $\mathcal{I}$. As a consequence, the set
$Y=\{C_{j}\in C\mid j\notin\pi(C_{j})\}$ forms an exact cover of
$X$. \hfill $\square$
\end{proof}

\smallskip
We remark that an analogous hardness result holds for \MINSUM\ if the trees are not required to contain all vertices of $V$, and in fact, for the case that none of the trees contains all vertices of $V$.\footnote{This can be proven by  introducing a distinct  dummy item for each agent and inserting it as a leaf, together with an arc from the root to that leaf, in the agent's preference graph.}

\medskip
Given the negative results derived even for quite simple preference graphs, a natural step to get closer to the boundary between ``hard'' and ``easy'' cases is a restriction on the number of agents.
But even for only two agents, we can derive an NP-completeness result for both objective functions.

\begin{theorem}\label{thm:twokids-hard}
{\MINSUM} and {\MINMAX} are NP-complete, even if the number of agents $k=2$ and the two sets of items are the same. 
\end{theorem}

\begin{proof}
To show the NP-hardness, we reduce from the $3$-SAT problem, which is known to be NP-complete~\cite{cook1971complexity}.

Let $\phi$ be a $3$-SAT formula with $n$ variables $x_1, \dots, x_n$ and $m$ clauses, each containing exactly $3$ literals.
We now construct instances of {\MINSUM} and {\MINMAX} using $K=\{1,2\}$, the same set $V$ of items and 
preference graphs $G_1 = (V, A_1)$, $G_2 = (V, A_2)$.
For each variable $x_i$ for $i=1,\dots,n$ we add two variable items $v_i, \bar{v}_i$
and a dummy item $u_i$ to $V$. For each clause $j=1,\dots,m$ we add a clause item $c_j$ to $V$. For each $j=1,\dots,m$ we add the arc $(v_i,c_j)$ to $A_1$ if the literal $x_i$ appears in the $j$-th clause and add the arc $(\bar{v}_i,c_j)$ to $A_1$ if the literal $\bar{x}_i$ appears in the $j$-th clause. For each $i=1,\dots,n$ we add the arcs $(v_i,u_i)$ and $(\bar{v}_i, u_i)$ to $A_2$. See Figure~\ref{fig:twokid-hard} for an illustration of the construction used in the reduction.
\begin{figure}[H]
	\centering
	\begin{tikzpicture}[vertex/.style={inner sep=2pt,draw,circle},scale=0.7]

	\begin{scope}[xshift=0cm]
	
	\node[vertex, label=above:$ v_1 $] (v1) at (0,3) {};
	\node[vertex, label=above:$ \bar{v}_1 $] (vv1) at (1,3) {};
	
	\node[vertex, label=above:$ v_2 $] (v2) at (3*1+0,3) {};
	\node[vertex, label=above:$ \bar{v}_2 $] (vv2) at (3*1+1,3) {};

	\node[vertex, label=above:$ v_3 $] (v3) at (3*2+0,3) {};
	\node[vertex, label=above:$ \bar{v}_3 $] (vv3) at (3*2+1,3) {};

	\node[vertex, label=above:$ v_4 $] (v4) at (3*3+0,3) {};
	\node[vertex, label=above:$ \bar{v}_4 $] (vv4) at (3*3+1,3) {};
	
	\node[vertex, label=below:$ c_1 $] (c1) at (3,0.5) {};
	\node[vertex, label=below:$ c_2 $] (c2) at (8,0.5) {};
	
	\node[vertex, label=above:$ u_1 $] (u1) at (13,2) {};
	\node[vertex, label=above:$ u_2 $] (u2) at (14,2) {};
	\node[vertex, label=above:$ u_3 $] (u3) at (13,1) {};
	\node[vertex, label=above:$ u_4 $] (u4) at (14,1) {};
	
	\path[->] (v1) edge (c1);
	\path[->] (vv3) edge (c1);
	\path[->] (vv4) edge (c1);
	
	\path[->] (vv1) edge (c2);
	\path[->] (v2) edge (c2);
	\path[->] (vv4) edge (c2);

	 \node [above of=v3] {\parbox{0.3\linewidth}{\subcaption{Graph $G_1$ of agent $1$}}};
	\end{scope}

		\begin{scope}[xshift=0cm,yshift=-5.5cm]
    \node[vertex, label=above:$ v_1 $] (v1) at (0,3) {};
	\node[vertex, label=above:$ \bar{v}_1 $] (vv1) at (1,3) {};
	
	\node[vertex, label=above:$ v_2 $] (v2) at (3*1+0,3) {};
	\node[vertex, label=above:$ \bar{v}_2 $] (vv2) at (3*1+1,3) {};

	\node[vertex, label=above:$ v_3 $] (v3) at (3*2+0,3) {};
	\node[vertex, label=above:$ \bar{v}_3 $] (vv3) at (3*2+1,3) {};

	\node[vertex, label=above:$ v_4 $] (v4) at (3*3+0,3) {};
	\node[vertex, label=above:$ \bar{v}_4 $] (vv4) at (3*3+1,3) {};
	
	\node[vertex, label=below:$ u_1 $] (u1) at (0.5,1) {};
	\node[vertex, label=below:$ u_2 $] (u2) at (3.5,1) {};
	\node[vertex, label=below:$ u_3 $] (u3) at (6.5,1) {};
	\node[vertex, label=below:$ u_4 $] (u4) at (9.5,1) {};
	
	\node[vertex, label=above:$ c_1 $] (c1) at (13,2) {};
	\node[vertex, label=above:$ c_2 $] (c2) at (14,2) {};
	
	\path[->] (v1) edge (u1);
	\path[->] (vv1) edge (u1);
	
	\path[->] (v2) edge (u2);
	\path[->] (vv2) edge (u2);
	
	\path[->] (v3) edge (u3);
	\path[->] (vv3) edge (u3);
	
	\path[->] (v4) edge (u4);
	\path[->] (vv4) edge (u4);

	       \node [above of=v3] {\parbox{0.3\linewidth}{\subcaption{Graph $G_2$ of agent $2$}}};
	\end{scope}
	
	\end{tikzpicture}\caption{Example of the reduction in Theorem~\ref{thm:twokids-hard} for the $3$-SAT instance given by the formula $(x_1 \vee \bar{x}_3 \vee \bar{x}_4) \wedge (\bar{x}_1 \vee x_2 \vee \bar{x}_4)$.}\label{fig:twokid-hard}
\end{figure}
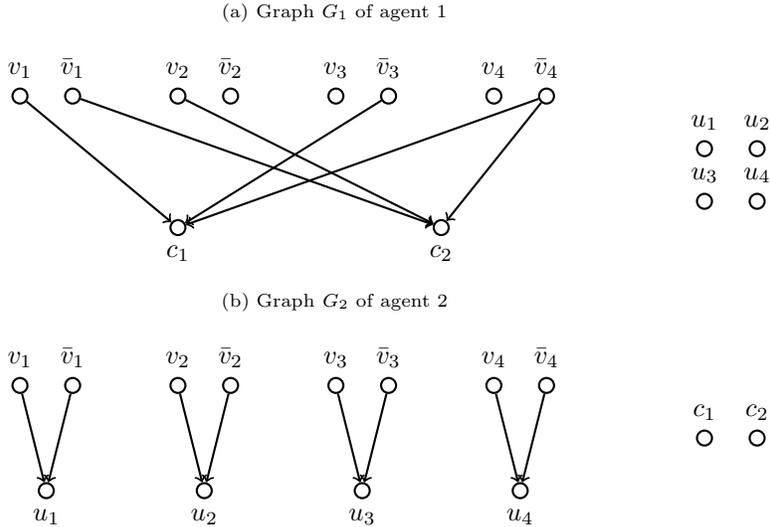
We claim that the following three statements are equivalent:
\begin{enumerate}[(1)]
    \item $\phi$ is satisfiable;
    \item there exists an assignment of items to agents in this instance such that the total dissatisfaction is at most~$2n$;
    \item there exists an assignment of items to agents in this instance such that the maximum dissatisfaction of any of the two agents of this instance is at most~$n$.
\end{enumerate}
First, observe that $2n$ is a lower bound on the total dissatisfaction of this instance
and $n$ is a lower bound on the maximum dissatisfaction of the two agents. This follows by the fact 
that the items $v_i, \bar{v}_i$ for $i=1,\dots,n$ are roots in both $G_1$ and $G_2$ and hence each such pair leads to a dissatisfaction of $1$ for at least one of the two agents.

Next, assume (1) holds and we are given a satisfying assignment for $\phi$. We construct an assignment of items to agents with 
total dissatisfaction $2n$, in which each of the two agents has dissatisfaction $n$. For each $i=1,\dots,n$ if $x_i = \textsf{true}$ we assign $v_i$ to agent $1$ and $\bar{v}_i$ to agent $2$. Otherwise, if $x_i = \textsf{false}$ we assign $\bar{v}_i$ to agent $1$ and $v_i$ to agent $2$. For all $j=1,\dots,m$ the clause item $c_j$ is assigned to agent $2$. For all $i=1,\dots,n$ and the dummy items $u_i$ are assigned to agent $1$. Note that in $G_1$ all the clause items $c_j$ are dominated by the respective variable item of the literal that satisfies the $j $-th clause. Also, all dummy items are assigned to agent $1$ and are dominated in $G_2$. Based on this it is easy to verify that the dissatisfaction of both agents is $n$ and the total dissatisfaction is $2n$. Hence (1) implies both (2) and (3).

For the converse direction, first observe that point (3) implies (2). Hence, we assume that (2) holds and we are given an assignment of items to the two agents such that the total dissatisfaction is $2n$. 
Note that for all $j=1,\dots,m$, the clause item $c_j$ is assigned to agent $2$, since otherwise $c_j$ cannot be dominated in $G_2$ and hence the total dissatisfaction must be at least $2n+1$, a contradiction. 
Also, for each $i=1,\dots,n$ at least one of the two items $v_i, \bar{v}_i$ must be assigned to agent $2$, since otherwise, the dummy item $u_i$ is not dominated in either $G_1$ or in $G_2$ and hence contributes $+1$ to the total dissatisfaction which then is at least $2n+1$, a contradiction.
Hence for each $i=1,\dots,n$ at most one of the two items, $v_i$ or $\bar{v}_i$, is assigned to agent $1$.
If $v_i$ is assigned to agent $1$ we set $x_i = \textsf{true}$. Otherwise we set $x_i = \textsf{false}$.
Note, that for all $j=1,\dots,m$, the clause item $c_j$ must be dominated by some variable item $v_i$ or $\bar{v}_i$ assigned to agent $1$. 
Otherwise, the clause item contributes  $+1$ to the dissatisfaction of one of the two agents, leading to a total dissatisfaction of at least $2n+1$, a contradiction. 
By our assignment of $x_1,\dots,x_n$ the corresponding literal in the $j$-th clause is set to $\textsf{true}$. Hence, $x_1,\dots,x_n$ is a satisfying assignment of $\phi$ and (3) follows.
\qed
\end{proof}

The construction of the proof allows a slightly stronger formulation of the result of Theorem~\ref{thm:twokids-hard}.
Considering that $3$-SAT is NP-complete even if every literal appears in exactly two clauses, see~\cite{Berman2003}, we get the following statement. 
\begin{remark}\label{thm:twokids-hard-degree}
For $k=2$ agents {\MINSUM} and {\MINMAX} are NP-complete, even if the preference graphs have out-degree at most two, in-degree at most three, no directed path of length at least two, and the two sets of items are the same.
\end{remark}

Deriving a positive counterpart to this negative result, we continue to consider the case of two agents and look for simple graph classes permitting polynomial-time solutions.
We succeed by showing that {\MINSUM} can be solved in polynomial time when the two preference graphs are collections of out-stars.
This can be compared to the construction given in the proof of Theorem~\ref{thm:twokids-hard} and specified in Remark~\ref{thm:twokids-hard-degree}, where NP-completeness was stated for preferences represented by restricted bipartite underlying graphs.

Later, in Theorem~\ref{th:constjunc}, we will consider preference structures where the number of all junction vertices (i.e., vertices with in- or out-degree greater than $1$) is constant. 
However, this does not cover the result of the subsequent Theorem~\ref{th:B-stars-2}, where we allow an arbitrary number of out-stars.

\medskip
Let the set of two agents be $K = \{i,j\}$ and let $G_i$ and $G_j$ be corresponding preference graphs.
We will call each vertex belonging to the set $(V(G_i)\cup V(G_j))\setminus(V(G_i)\cap V(G_j))$ a {\em{personal item}}.
For non-personal items $v$ there are three possibilities: (i) $v$ is a root or a leaf in both graphs;
(ii) ${\it pred}_i(v)=\emptyset$ and ${\it pred}_j(v)\neq\emptyset$; 
(iii) ${\it pred}_j(v)=\emptyset$ and ${\it pred}_i(v)\neq\emptyset$.
In cases (ii) and (iii) we call $v$ a \textit{one-root item} of (an out-star in) $G_i$, resp.~$G_j$. 

\begin{lemma}\label{Lem:out-stars}
For $k=2$ when both graphs $G_i$ are disjoint unions of out-stars,
there exists an optimal allocation $\pi^*$ for {\MINSUM} with the following properties:
\begin{enumerate}
    \item each personal item is assigned to the corresponding agent;
    \item each one-root item $v$ of $G_i$ is assigned to agent $i$, and for every $(v,u) \in A(G_i)$, $u$~is assigned to the unique agent $j$ in $K\setminus\{i\}$ provided $u$ is not a personal item.
   
\end{enumerate}
\end{lemma}

\begin{proof}
Without loss of generality assume $K=\{1,2\}$. Given an allocation $\pi$ of this problem, we can construct an allocation $\pi'$ with $\sum_{i\in K}\delta_{\pi'}(i)\leq \sum_{i\in K}\delta_\pi(i)$ by applying one of the following steps:

\begin{enumerate}
    \item If there exists a personal item $v$ in $G_1$ and in $\pi$ it was not assigned to agent $1$, then let $\pi'$ be as $\pi$ except we assign $v$ to agent $1$. We have $\delta_{\pi'}(2) = \delta_\pi(2)$ and $\delta_{\pi'}(1)\leq \delta_{\pi}(1)$. The same argument applies if we have a personal item in $G_2$. Hence the total dissatisfaction of $\pi'$ is not greater than the one of~$\pi$.
    \item Suppose an item $v$ is a one-root item of an out-star in $G_1$ and that $v$ was not assigned to agent $1$ in $\pi$.
    Let $\pi'$ be the same assignment as $\pi$, with the exception that we assign item $v$ to agent $1$, and we assign all out-neighbors 
    of $v$ that are not personal items in $G_1$ to agent $2$. We have $\delta_{\pi'}(2)\leq \delta_\pi(2)+1$ and $\delta_{\pi'}(1)\leq \delta_\pi(1)-1$.
    The same argument applies if we have a one-root item in $G_2$.
    Hence the total dissatisfaction of $\pi'$ is not greater than the one of~$\pi$. 
\end{enumerate}

Applying these steps to an allocation $\pi$ repeatedly and setting $\pi = \pi'$ after each application of a step leads to an allocation that satisfies properties $1$ and $2$.

Note that no matter which allocation $\pi$ we start with, the vertex sets considered in step $1$ and step $2$ depend only on the graphs $G_1$ and $G_2$, and are thus unique, i.e., lead to the same assignment of personal items and one-root items together with their leaves.
This implies that we can transform any allocation into an allocation compliant with the two properties. In particular, there exists an optimal allocation with the stated properties.
\qed
\end{proof}

Let us call {\em{preassignment}} the allocation compliant with Lemma~\ref{Lem:out-stars} that is obtained from an empty assignment.

\begin{theorem}\label{th:B-stars-2}
For $k=2$ when both graphs $G_i$ are disjoint unions of out-stars, we can find an allocation of minimum total dissatisfaction in an instance $\mathcal{I}$ of {\MINSUM} in polynomial time.
\end{theorem}

\begin{proof}
We reduce this problem to the maximum weight independent set problem on a bipartite graph, which is polynomial-time solvable (see, e.g.,~\cite{FF06}).

Without loss of generality assume $K=\{1,2\}$. Given an instance $(G_1,G_2)$ of {\MINSUM}, start by the preassignment $\pi_0$, and let $G_1'$ and $G_2'$ be the subgraphs of $G_1$ and $G_2$, respectively, induced by the yet unassigned items (note that $V(G'_1)=V(G'_2)$ and that the graphs $G_1'$ and $G_2'$ have the same sets of roots, as well as the same sets of leaves).
Next, construct an undirected graph $G'$ by taking the disjoint union of the underlying undirected graphs of $G_1'$ and $G_2'$ and joining by an edge each pair of vertices corresponding to the same item. Formally, we set $V(G')=V(G_1') \times \{1,2\}$ and $E(G')=
\{\{(v,i),(u,i)\} \mid (u,v) \in A(G_i'), i \in \{1,2\}\} \cup \{\{(v,1),(v,2)\} \mid v \in V(G_1')\}$.
It is easy to see that the vertices of $G'$ that correspond to the roots of $G'_1$ together with the leaves of $G'_2$ form an independent set in $G'$.
Symmetrically, the vertices in $G'$ that correspond to the roots of $G'_2$ together with the leaves in $G'_1$ also form an independent set.
Hence, $G'$ is a bipartite graph.

To every vertex $(v,i) \in V(G')$, we assign a weight one larger than the out-degree of $v$ in $G_i$ minus the number of successors of $v$ that were already preassigned. % to some agent.
See Fig.~\ref{fig:Bipartite-construction} for an example.
This weight represents the additional satisfaction obtained by $i$ if it is assigned item $v$.

\begin{figure}
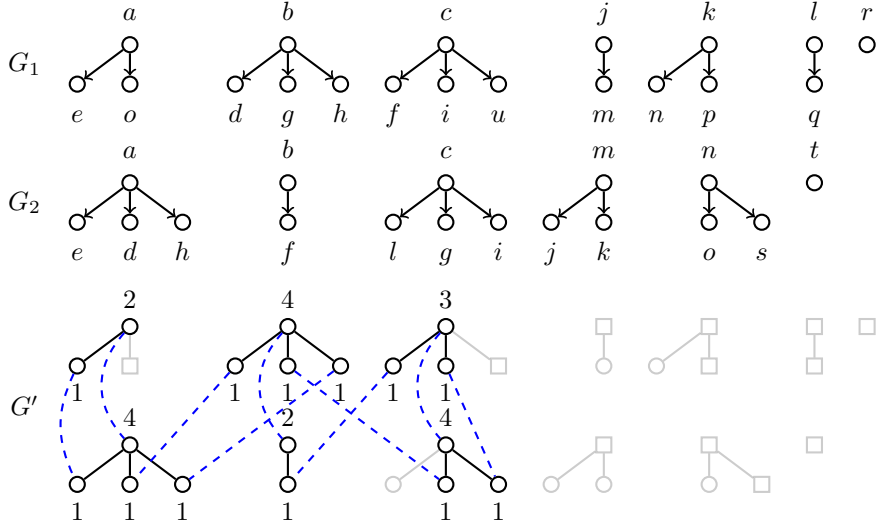

\centering
\pictureoutstar
\caption{An example of how to construct $G'$ from two given graphs $G_1$ and $G_2$. Graphs $G_1$ and $G_2$ have labeled vertices while in $G'$ only the weight of the vertices is displayed. Note that in $G'$ the upper part represents 
% entries from 
$G_1'$ and the lower part of the graph represents 
% entries from 
$G_2'$. Furthermore, the square vertices that appear transparent were assigned to the corresponding agent by the preassignment, whereas the round transparent vertices correspond to items that were given to the other agent.}\label{fig:Bipartite-construction}
\end{figure}

An independent set $S$ in $G'$ induces an allocation of items $\pi$ for $(G_1,G_2)$, as follows. We begin with the preassignment $\pi=\pi_0$.
Furthermore, if $(v,i)\in S$, then we assign item $v$ to agent $i$ in $\pi$. The edges of the form $\{\{(v,1),(v,2)\} \mid v \in V(G_1')\}$ prevent the assignment of the same item twice.
Conversely, an allocation $\pi$ of $(G_1,G_2)$, such that no two items are assigned to the same agent are comparable (i.e., root and leaves of a same out-star of agent $i$ are not both assigned to $i$), induces an independent set $S$ in $G'$.

We now show that a maximum weight independent set $S$ in $G'$ 
corresponds to
%  induces 
an optimal allocation of items $\pi$ for $(G_1,G_2)$, and vice versa.
Note that the weight of each vertex $(v,i)$, by construction, 
equals the number of items which were not preassigned that are dominated by $v$ in $G_i$.
Observe that if no vertex of $G'$ is assigned to any agent, then the total dissatisfaction of the corresponding allocation $\pi_0$ would equal
% $\pi$ would equal
the sum of the weights of the vertices of $G'$ corresponding to roots in $G'_1$ and $G'_2$.
Take any independent set $S$ of $G'$ and consider the corresponding assignment $\pi$.
Note also that $\sum_{i \in K}\delta_{\pi_0}(i)$ is also equal to $w(S)+\sum_{i \in K}\delta_{\pi}(i)$ (this can easily be shown by induction on $|S|$).
% \sum_{i \in K}\delta_{\pi}(i) = \sum_{i \in K}\delta_{\pi_0}(i)$ - w(S)
From here we can see that as we maximize $w(S)$ we also minimize 
$\sum_{i \in K}\delta_{\pi}(i)$.
Hence, $S$ is optimal for the maximum weight independent set problem if and only if $\sum_{i \in K}\delta_{\pi}(i)$ is optimal for {\MINSUM}.

The polynomial reduction from {\MINSUM} to bipartite maximum weight
independent set is now complete. Thus, it follows that {\MINSUM} for two agents with disjoint unions of out-stars is also polynomial.
\qed
\end{proof}

\section{Directed matchings as preference graphs}
\label{sec:matching}

After the strikingly negative results of Section~\ref{sec:hardness}, where it was shown that even elementary graphs, such as out-stars and out-trees, imply the NP-completeness of our two problems, we now consider 
the basic graph structure of directed matchings.
Indeed, the pairwise comparison of two items with no connection to any other items seems to be one of the most basic possibilities of considering any preferences at all.
Also, in decision science, the pairwise comparison of options constitutes the elementary building block for multi-criteria decision making methods, such as outranking methods~\cite{mcda-book}.

We exhibit an interesting difference between the two objectives.
While {\MINSUM} is shown to be polynomially solvable if all $G_i$ are directed matchings, the same situation turns out to be still NP-complete for {\MINMAX}.
Nonetheless, we obtain a positive result for the Min-Max objective for the special case of $k=2$ agents.

\begin{theorem}\label{th:minSUM-matching}
When each graph $G_i=(V_i, A_i)$ is a directed matching {\MINSUM} can be solved in polynomial time.
\end{theorem}

\begin{proof}
To solve {\MINSUM} we will compute a maximum weight matching on an auxiliary undirected \newSL{bipartite} graph $H$. 
Its vertex set $V(H)= X\cup Y\cup Z$ 
consists of a vertex for every vertex in $V_i$, 
i.e., $X=\{x_j^i\mid i\in K, j\in V_i\}$, 
a vertex for every arc in $A_i$, i.e., 
$Y=\{y_a^i\mid i\in K,a\in A_i\}$, 
and a vertex for every item, i.e., $Z=\{z_\ell\mid \ell\in V\}$.
The edge set $E(H)= S\cup T$ 
contains edges connecting every item vertex in $Z$ with all its copies in $X$, i.e., $S=\{\{x_j^i,z_j\} \mid i\in K, j\in V_i\}$, 
and edges connecting the two endpoints of a matching arc in $A_i$ with the corresponding vertex in $Y$, i.e., 
for each agent $i\in K$, and each arc $(a,a') \in A_i$, set $T$ contains edges $\{y_{(a,a')}^i,x_a^i\}$ and $\{y_{(a,a')}^i,x_{a'}^i\}$.
% $$T= \bigcup_{a_\ell^i\in Y, (q,q') \in A_i, i\in K}\{\{a_{(q,q')}^i,p_q^i\}, \{a_{(q,q')}^i,p_{q'}^i\}\}\,.$$
We claim that every matching $M$ in $H$ implies a feasible allocation of items to the agents by assigning item $j$ to agent $i$ if $e=\{x_j^i,z_j\}\in M$. 
Since there can be at most one edge in $M$ joining a vertex $z_j$ in $Z$ to a vertex $x_j^i$ in $X$, every item is allocated at most once.
To avoid that both endpoints of an arc in $A_i$ are allocated to $i$,
the edges in $T$ are assigned a very high weight.
Then, every maximum weight matching will contain one of the two edges in $T$ incident with a vertex $y^i_a$ in $Y$, which forbids that the other endpoint in $X$ corresponds to an item allocated to $i$.

The following weights are assigned to each $e\in E(H)$:
\[w(e)=\begin{cases}
        %0, & \text{if } e\in R \\
        1 & \text{if } e= \{x_j^i,z_j\}\in S \text{ and } j \text{ is the head of an edge in $A_i$,}\\
        2 & \text{if } e=\{x_j^i,z_j\}\in S \text{ and } j \text{ is the tail of an edge in $A_i$,}\\
        2|V| & \text{if } e\in T.
  \end{cases}
\]
      
The weights on the edges in $S$ correspond to the number of vertices that each vertex from $V_i$ dominates in $G_i$. 
Hence, a maximum weight matching will 
%contain exactly $|T|/2$ edges of weight $2n$ and 
maximize the total satisfaction and thus minimize the total dissatisfaction.
\newSL{It holds that the maximum total satisfaction is equal to $w(M) - 2|T| \cdot |V|$ and hence the 
minimum total dissatisfaction is equal to $\sum_{i \in K} |V_i| + 2 |T| \cdot |V| - w(M)$.}
\qed 
\end{proof}

\begin{theorem}\label{th:hard-A-matching}
{\MINMAX} is NP-complete, even if each graph $G_i$ is a directed matching.
\end{theorem}

\begin{proof} We again reduce from 3X3C. Given an instance $\mathcal{J}$ of 3X3C with a set $X$ of elements and a collection $C=\{C_{1},\ldots,C_{p}\}$ of $3$-element subsets of $X$, let $\ell=\frac{4p}{3}$. 
We may assume without loss of generality that $p\geq 6$ and thus $\ell \geq 8$. We construct an instance $\mathcal{I}$ of the  {\MINMAX}
by \newwSL{introducing the items $V=X\cup\{h_{j}\mid1\leq j\leq  \ell +1\}\cup\{j,b_j^0,b_j^1,e_j\mid 1\leq j \leq p\} \cup \{a_{j}\mid1\leq j\leq  \ell +1\} $, and the set $K$ of agents made up of agents $D_{j}$ for $1\leq j \leq \ell +1$, agent $F$, and the agents $B_j, C_j$ for  
$1\leq j \leq p$; their graphs are displayed in Figure~\ref{fig:MinMax-Matching} where w.l.o.g.\ we assume set  $C_j=\{x,y,z\}$.
We ask whether there is an allocation with dissatisfaction of at most $\ell$ per agent.
}
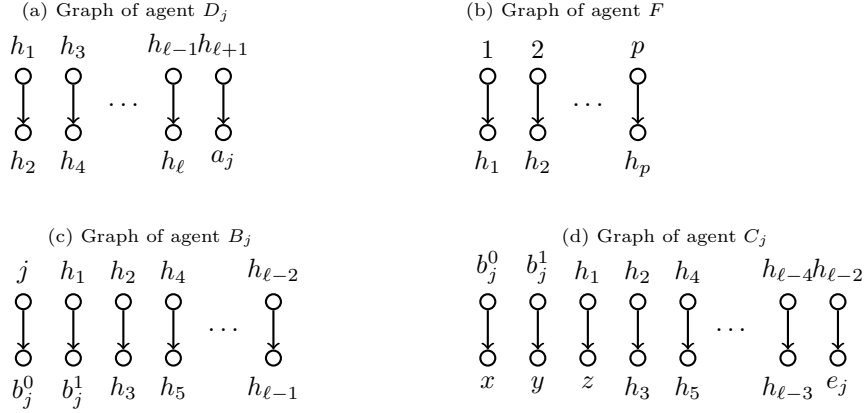
\begin{figure}[htbp]
	\centering
	\begin{tikzpicture}[vertex/.style={inner sep=2pt,draw,circle},yscale=0.75,xscale=0.95]
	
	\begin{scope}[xshift=-4.5cm]
	\node[vertex, label=above:$ h_1 $] (1) at (2.5,5) {};
	\node[vertex, label=below:$ h_2 $] (2) at (2.5,4) {};
	\node[vertex, label=above:$ h_3 $] (3) at (3.2,5) {};
	\node[vertex, label=below:$ h_4 $] (4) at (3.2,4) {};
	\node[] (5) at (3.9,4.5) {$\ldots$};
	\node[vertex, label=above:$ h_{\ell-1} $] (6) at (4.6,5) {};
	\node[vertex, label=below:$ h_{\ell} $] (7) at (4.6,4) {};
	\node[vertex, label=above:$ h_{\ell+1} $] (8) at (5.3,5) {};
	\node[vertex, label=below:$ a_j $] (9) at (5.3,4) {};
	    \path[->] (3) edge (4);
	    \path[->] (6) edge (7);
	    \path[->] (8) edge (9);
	    \path[->] (1) edge (2);
	\node at (3.95,6) {\parbox{0.3\linewidth}{\subcaption{Graph of agent $D_j$}\label{subfig:a}}};
	\end{scope}

	\begin{scope}[xshift=2cm]
	\node[vertex, label=above:$ 1 $] (1) at (2.5,5) {};
	\node[vertex, label=below:$ h_{1} $] (2) at (2.5,4) {};
	\node[vertex, label=above:$ 2 $] (3) at (3.2,5) {};
	\node[vertex, label=below:$ h_{2} $] (4) at (3.2,4) {};
	\node[] (5) at (3.9,4.5) {$\ldots$};
	\node[vertex, label=above:$ p $] (6) at (4.6,5) {};
	\node[vertex, label=below:$ h_{p} $] (7) at (4.6,4) {};
	    \path[->] (3) edge (4);
	    \path[->] (6) edge (7);
	    \path[->] (1) edge (2);
    \node at (3.6,6) {\parbox{0.3\linewidth}{\subcaption{Graph of agent $F$}\label{subfig:c}}};
	\end{scope}% 	
	
	\begin{scope}[xshift=-4.5cm,yshift=-4cm]
	\node[vertex, label=above:$ j $] (10) at (2.5,5) {};
	\node[vertex, label=below:$ b_j^0 $] (11) at (2.5,4) {};
	\node[vertex, label=above:$ h_1 $] (1) at (3.2,5) {};
	\node[vertex, label=below:$ b_j^1 $] (2) at (3.2,4) {};
	\node[vertex, label=above:$ h_2 $] (3) at (3.9,5) {};
	\node[vertex, label=below:$ h_3 $] (4) at (3.9,4) {};
	\node[] (5) at (5.3,4.5) {$\ldots$};
	\node[vertex, label=above:$ h_4 $] (6) at (4.6,5) {};
	\node[vertex, label=below:$ h_5 $] (7) at (4.6,4) {};
	\node[vertex, label=above:$ h_{\ell-2} $] (8) at (6,5) {};
	\node[vertex, label=below:$ h_{\ell-1} $] (9) at (6,4) {};
	    \path[->] (3) edge (4);
	    \path[->] (6) edge (7);
	    \path[->] (8) edge (9);
	    \path[->] (10) edge (11);
    	\path[->] (1) edge (2);
	\node at (4.3,6) {\parbox{0.3\linewidth}{\subcaption{Graph of agent $B_j$}\label{subfig:d}}};
	\end{scope}

	\begin{scope}[xshift=2cm,yshift=-4cm]
	\node[vertex, label=above:$ b_j^0 $] (10) at (2.5,5) {};
	\node[vertex, label=below:$ x $] (11) at (2.5,4) {};
	\node[vertex, label=above:$ b_j^1 $] (1) at (3.2,5) {};
	\node[vertex, label=below:$ y $] (2) at (3.2,4) {};
	\node[vertex, label=above:$ h_1 $] (3) at (3.9,5) {};
	\node[vertex, label=below:$ z $] (4) at (3.9,4) {};
	\node[] (5) at (5.9,4.5) {$\ldots$};
	\node[vertex, label=above:$ h_2 $] (6) at (4.6,5) {};
	\node[vertex, label=below:$ h_3 $] (7) at (4.6,4) {};
	\node[vertex, label=above:$ h_4 $] (12) at (5.3,5) {};
	\node[vertex, label=below:$ h_5 $] (13) at (5.3,4) {};
	\node[vertex, label=above:$ h_{\ell-4} $] (8) at (6.7,5) {};
	\node[vertex, label=below:$ h_{\ell-3} $] (9) at (6.7,4) {};
	\node[vertex, label=above:$ h_{\ell-2} $] (14) at (7.4,5) {};
	\node[vertex, label=below:$ e_j $] (15) at (7.4,4) {};
	    \path[->] (3) edge (4);
	    \path[->] (6) edge (7);
	    \path[->] (8) edge (9);
	    \path[->] (10) edge (11);
	    \path[->] (12) edge (13);
	    \path[->] (14) edge (15);
    	\path[->] (1) edge (2);
    \node at (5,6) {\parbox{0.3\linewidth}{\subcaption{Graph of agent $C_j$}\label{subfig:d}}};
	\end{scope}
	
	\end{tikzpicture}\caption{Graphs of the agents in the proof of Theorem~\ref{th:hard-A-matching}.}\label{fig:MinMax-Matching}
\end{figure}
  Observe that $\ell$ is even since $p$ is a multiple of $3$. 
  \newwSL{Using this construction it can be shown that $C$ contains an exact cover of $X$  if and only if $\mathcal{I}$
 admits an allocation $\pi$ with $\max_{i \in K} \diss_{\pi}(i) \leq \ell$.}
  
We claim that $C$ contains an exact cover of $X$  if and only if $\mathcal{I}$ admits an allocation $\pi$ with $\max_{i \in K} \diss_{\pi}(i) \leq \ell$. \\
 
 Assume first that $\pi$ is such an allocation. In order to respect the bound $\ell$, each of the agents $D_j$, $1\leq j \leq \ell +1$ must receive at least one of $h_1,\ldots,h_{\ell+1}$ under allocation $\pi$. 
 %As a consequence, agent $A$ has to receive all the items $d_{\ell+2},\ldots, d_{2\ell+1}$. 
 Hence, all of the items $h_j$ are already allocated. 
 Therefore, agent $F$ has to receive at least $\frac{p}{3}$ items of $1,\ldots, p$ because otherwise her dissatisfaction would exceed $\frac{2p}{3}\cdot 2=\ell$.  
 Now, if item $j$ is allocated to agent $F$, agent $B_j$ needs to receive both $b_j^0$ and $b_j^1$; thus, agent $C_j$ (representing set $C_j=\{x,y,z\}$ in instance $\mathcal{J}$ of 3X3C) needs to receive the items $x,y,z$, and $e_j$. Since (at least) $\frac{p}{3}$ of the items $j$ are allocated to agent $F$, this means that there must be exactly $\frac{p}{3}$ agents $C_j$ who receive all three items that make up the set of the same label in instance $\mathcal{J}$ of 3X3C. Hence, the  respective sets form an exact cover in instance $\mathcal{J}$ of 3X3C.
 
Assume now that $Y$ is an exact cover of $X$ in $C$. We derive an allocation $\pi$ as follows: 
 \begin{itemize}
  \item for each $1\leq j \leq \ell +1$, assign to agent $D_j$ items $h_j$ and $a_j$;
  %\item assign to agent $A$ the items $d_{\ell+2},\ldots, d_{2\ell+1}$;
  \item for all items $j$ with $C_j \in Y$: give $j$ to agent $F$, $b_j^0$ and $b_j^1$ to agent $B_j$, and give item $e_j$ plus the three items corresponding to the elements that make up set $C_j$ to agent $C_j$;
  \item for all items $j$ with $C_j \notin Y$: give $j$ to agent $B_j$, and give the items $b_j^0, b_j^1, e_j$ to agent $C_j$.
 \end{itemize}
Since agent $F$ receives $\frac{p}{3}$ of the items $1, \ldots, p$ her dissatisfaction is exactly $\frac{2p}{3}\cdot 2=\ell$. It is not difficult to verify that the dissatisfaction of the remaining agents is at most $\ell$ as well. 
\qed 
\end{proof}

\medskip
Given the negative result of the above theorem we give a complementing positive result for {\MINMAX} below. 
Namely, if we again restrict the number of agents to two, we have a positive counterpart to Theorem~\ref{th:hard-A-matching}.

In the proof of the following theorem we will make use of the notation $S + v$ for the sum of a set $S\subseteq \mathbb{Z}^2$ of ordered pairs with an ordered pair $v\in \mathbb{Z}^2$, defined as $S+v := \{x + v \colon x \in S\}$.

\begin{theorem}\label{th:A-matching-2}
When $k=2$ and both preference graphs are directed matchings, {\MINMAX} can be solved in polynomial time.
\end{theorem}

\begin{proof}
Let $K = \{1,2\}$.
First note that since $G_1$ and $G_2$ are directed matchings, the underlying undirected graph of $G = G_1 \cup G_2$ (including possible multi-edges) is a collection of (agent~1)-(agent~2) alternating cycles (including cycles with two vertices) and paths (including single-edge paths). 
Hence, in the following proof we call the directed counterparts in $G$ of these cycles and paths also cycles and paths (even though they are not necessarily cycles and paths in the usual directed sense).

In the first part of the proof we will show how to obtain the set of all possible dissatisfaction profiles $(d_1,d_2)$ of the two agents for one path and then also for one cycle using a dynamic programming approach.
In the second part we will then show how to combine these sets of dissatisfaction profiles for all the paths and cycles to obtain all dissatisfaction profiles with respect to all items.

Given a path $P =(v_1, \dots, v_\ell)$, 
note that all vertices of $P$ must be contained in both $G_1$ and $G_2$ except $v_1$ and $v_\ell$.
%The only step when this is not the case is for $j=\ell$. 
%In this case the change in the dissatisfaction can be easily adapted.
For $j\in \{1,\ldots,\ell\}$, we denote by $d_P(j,i)$ the set of dissatisfaction profiles of the two agents assuming the graph is only the subpath $P_j$ of $P$ consisting of the first $j$ vertices and the $j$-th vertex is assigned to agent $i \in \{1,2\}$.
Note that if $(v_j,v_{j+1})$ is an arc of $G_i$, then this arc is not yet considered here, i.e., vertex $v_j$ is treated like an isolated vertex for agent $i$ at this point.
% Note that a potential arc between $v_j$ and $v_{j+1}$ in $G_i$ is not yet considered here, i.e., if $(v_j,v_{j+1})\in A_i$, then vertex $v_j$ is treated like an isolated vertex for agent $i$ at this point.
Now it is easy to see that $d_P(1,i) = \{(d_1,d_2)\}$,
where for $q=1,2$ we set $d_q = 1$ if $i \neq q$ and $v_1 \in V(G_q)$, otherwise $d_q = 0$.
Using a dynamic programming technique we continue to process the vertices of $P$ in increasing order of indices and compute all dissatisfaction profiles of the corresponding subpaths.
Formally, for $j>1$, given $d_P(j-1,1)$ and $d_P(j-1,2)$ we 
can compute $d_P(j,i)$ by adding $v_j$ to $P_{j-1}$ to obtain $P_j$.
To do so, we have to consider four cases:
the newly considered arc between $v_{j-1}$ and $v_j$ can be contained either in $G_1$ or in $G_2$ and can be directed either as $(v_{j-1}, v_j)$ or $(v_{j}, v_{j-1})$.
Suppose first that $j<\ell$, that is, $v_j$ is an internal vertex of $P$ (in which case $v_j$ is in both $G_1$ and $G_2$).
If the arc between $v_{j-1}$ and $v_j$ is in $G_1$ and equals $(v_{j-1}, v_j)$ we have that 
% $d_P(j,1) = (d_P(j-1,1) + (0,1)) \cup (d_P(j-1,2) + (0,1))$ and
$d_P(j,1) = (d_P(j-1,1) \cup d_P(j-1,2)) + (0,1)$ and
$d_P(j,2) =  d_P(j-1, 1) \cup (d_P(j-1,2)+(1,0))$.
If the arc is in $G_1$ but equals $(v_j, v_{j-1})$ we have that 
$d_P(j,1) = (d_P(j-1,1) + (0,1)) \cup (d_P(j-1,2) + (-1,1))$ and
$d_P(j,2) = (d_P(j-1,1) \cup d_P(j-1,2)) + (1,0))$.
% If the arc is in $G_1$ but equals $(v_j, v_{j-1})$ we have that 
% $d_P(j,1) = (d_P(j-1,1) + (0,1)) \cup (d_P(j-1,2) + (-1,1))$ and
% $d_P(j,2) = (d_P(j-1,1) + (1,0)) \cup (d_P(j-1,2) + (1,0))$.
The cases when the arc is in $G_2$ can be handled symmetrically.
For the last vertex $v_j=v_\ell$ the computation can be easily adapted.
For example, if the last arc is in $G_1$, then it suffices to replace any additive term of the form $(x,1)$ in the above formulas with term $(x,0)$ (since vertex $v_\ell$ does not belong to $G_2$, and hence the corresponding item cannot count towards dissatisfaction of agent $2$).

For a cycle $C = v_0, v_1, \dots, v_\ell$ one decides at the beginning to which agent $a$ the item of vertex $v_0$ is assigned to, and keeps track of this decision.
Formally, we consider the set of dissatisfaction profiles $d_C(a, j, i)$  for the two agents assuming the graph consists 
only of the subpath $C_j$ of $C$ with vertices $v_0, \dots, v_j$, where $v_0$ is assigned to $a$ and $v_j$ is assigned to $i$.
These sets can be computed in the same way as $d_P(j,i)$ above, and we can easily compute the dissatisfaction profiles of $C$ from $d_C(a,j,i)$, by considering the effect of the arc between $v_\ell$ and $v_0$, since all possible dissatisfaction profiles for all assignments of the two vertices are known.

After computing all sets of dissatisfaction profiles $d_P$ and $d_C$ for all paths and cycles we can compute the set $d_G$ of all dissatisfaction profiles  of the whole graph $G$.
Starting with $d_G=\emptyset$ we simply go through all paths and cycles one after the other. % (in arbitrary order).
For each of them we take every element of the corresponding set of dissatisfaction profiles and add it to every element of the current set $d_G$, resulting in an updated set $d_G$.

Note that all these dynamic programs can be computed in polynomial time, since for any instance of {\MINMAX} with $k$ agents the number of all possible dissatisfaction profiles is bounded by $(n+1)^k$.
Here we have $k=2$ and thus $|d_H| \leq (n+1)^2$ for all subgraphs $H$ considered in this proof.
\qed 
\end{proof}

\begin{remark}\label{rem:k2t2}
We point out that Theorem~\ref{th:A-matching-2} also follows from the more general Theorem~\ref{th:tw-minmax} (stated in  Section~\ref{sec:junction}). However, in contrast to the proof of  Theorem~\ref{th:tw-minmax}, the above proof of Theorem~\ref{th:A-matching-2} is based on elementary arguments and does not rely on the heavy machinery of Courcelle's theorem.
\end{remark}

\section{Paths as preference graphs}
\label{sec:special}

A classical preference structure is given by a total order of objects. 
In our setting, this corresponds to a path as a preference graph.
Allowing a subset of items to be ranked, the path can have arbitrary length.
For this case, both our objectives permit polynomial-time solutions.
For {\MINSUM} the result also applies if every agent can partition the items into incomparable subsets and gives a strict ordering, i.e., a path, for every subset.

\begin{theorem}\label{th:A-paths}
When each graph $G_i$ is a path, {\MINMAX} and {\MINSUM} can be solved in polynomial time.
\end{theorem}
\begin{proof}
For every vertex $v \in V(G_i)$, let $\ell_i(v)$ denote the number of predecessors of $v$ on the path $G_i$ (for the root vertex $r \in V(G_i)$ there is $\ell_i(r)=0$).
For $v\not\in G_i$, we set $\ell_i(v)=|V(G_i)|$.
Define the complete bipartite graph $H=(V, K)$ connecting items to agents with weights $w(v,i) = \ell_i(v)$ for $v\in V$ and $i \in K$.
Every maximum matching in $H$ corresponds to a feasible allocation of items to agents.
Note that for all graphs with $|V(G_i)|>k$, the path $G_i$ can be reduced to its top $k$ elements, since every agent will be assigned at most one item. 

The classical {\em Linear Sum Assignment Problem} (LSAP) asks for a maximum matching with minimum total weight in a bipartite graph.
Thus, LSAP on $H$ will also solve {\MINSUM}.

The {\em Linear Bottleneck Assignment Problem} (LBAP) seeks a maximum matching in a weighted bipartite graph such that the largest weight of a matching edge is as small as possible (see \cite[Sec.~6.2]{assign2012}, \cite{BAP97}).
Clearly, the optimal solution of LBAP on $H$ also minimizes the maximum dissatisfaction over all agents and thus solves {\MINMAX}.
\qed 
\end{proof}

\begin{coro}\label{th:sum-paths}
When each graph $G_i$ is a disjoint union of 
% consists of a collection of 
paths {\MINSUM} can be solved in polynomial time.
\end{coro}

\begin{proof}
If each $G_i$ consists of $k_i$ paths, it suffices to copy every agent $i$ into $k_i$ agents, each of them associated to exactly one of the paths.
Solving the resulting instance of {\MINSUM} with $\sum_i k_i$ agents as described in Theorem~\ref{th:A-paths} also solves the problem for the collection of paths.
\qed 
\end{proof}

Surprisingly, the straightforward generalization given in Corollary~\ref{th:sum-paths} does not carry over to the case of {\MINMAX}. 
On the contrary, we can show a strong negative result even for the special case where every $G_i$ consists of at most two paths.

\begin{theorem}\label{th:max-paths}
{\MINMAX} is NP-complete, even if each graph $G_i$ consists of at most two paths containing at most five items in total.
\end{theorem}

\begin{proof}
Clearly, {\MINMAX} belongs to NP.
To show the NP-hardness, we reduce from the $3$-SAT problem, which is known to be NP-complete~\cite{cook1971complexity}.

Let $\phi$ be a $3$-SAT formula with $n$ variables and $m$ clauses, each containing exactly $3$ literals.
We assume that every clause of $\phi$ contains literals that correspond to different variables.
For convenience, all the upper indices in the proof are assumed to be in the range $\{1,\ldots, m\}$.
We construct an instance $I$ of {\MINMAX} with a bound of $2$ for the maximum dissatisfaction
% maximum dissatisfaction $2$ 
as follows.
The set of items is $\{z,z',z''\} \cup \bigcup_{1 \leq i \leq n, 1 \leq j \leq m}\{v_i^j, \ell_i^j, \bar\ell_i^j\}$.
% Block agents
We create three agents $b$, $b'$ and $b''$ with same preference graph consisting of a unique path $(z,z',z'')$.
% Variable agents
Then, for each variable $x_i$, we create $m$ agents $a_i^1, \dots a_i^m$, who we call \emph{variable agents}, where the preference graph of each $a_i^j$ consists of two paths $(z, v_i^{j+1}, \ell_i^j)$ and $(v_i^j, \bar\ell_i^j)$.
% Clause agents
Now, for all $j \in \{1,\dots,m\}$, we consider the $j$th clause $(u_p \vee u_q \vee u_r)$ and create an agent $c_j$, who we call a \emph{clause agent}, with a path $(\tilde{\ell}_p^j,\tilde{\ell}_q^j,\tilde{\ell}_r^j)$ as preference graph where $\tilde{\ell}_p^j$ is the item $\ell_p^j$ if $u_p = x_p$ and $\tilde{\ell}_p^j$ is the item $\bar\ell_p^j$ if $u_p = \bar{x}_p$; the same holds for $\tilde{\ell}_q^j$ and $\tilde{\ell}_r^j$.
See Figure~\ref{fig:3sat construction} for a schematic representation of those preference graphs.
\begin{figure}
    \centering
	\begin{tikzpicture}[vertex/.style={inner sep=2pt,draw,circle},yscale=0.9,scale=0.74]
    \begin{scope}
    \node[vertex,label={right:$\mathstrut z$}] (z) at (0,2) {};
    \node[vertex,label={right:$\mathstrut z'$}] (z') at (0,1) {};
    \node[vertex,label={right:$\mathstrut z''$}] (z'') at (0,0) {};
    \path[->] (z) edge (z');
    \path[->] (z') edge (z'');
	\node at (0,3) {\parbox{0.3\linewidth}{\subcaption{Graph of the agents\\ $b$, $b'$ and $b''$}\label{subfig:agents bs}}};
    \end{scope}
    \begin{scope}[xshift=5.5cm]
    \node[vertex,label={left:$\mathstrut z$}] (z) at (0,2) {};
    \node[vertex,label={left:$\mathstrut v_i^{j+1}$}] (v+1) at (0,1) {};
    \node[vertex,label={left:$\mathstrut \ell_i^j$}] (ell) at (0,0) {};
    \node[vertex,label={right:$\mathstrut v_i^{j}$}] (v) at (1,1) {};
    \node[vertex,label={right:$\mathstrut \bar\ell_i^j$}] (ellbar) at (1,0) {};
    \path[->] (z) edge (v+1);
    \path[->] (v+1) edge (ell);
    \path[->] (v) edge (ellbar);
	\node at (0,3) {\parbox{0.3\linewidth}{\subcaption{Graph of the variable agent $a_i^j$}\label{subfig:agents bs}}};
    \end{scope}
    \begin{scope}[xshift=11cm]
    \node[vertex,label={right:$\mathstrut \bar{\ell}_1^j$}] (z) at (0,2) {};
    \node[vertex,label={right:$\mathstrut \ell_2^j$}] (z') at (0,1) {};
    \node[vertex,label={right:$\mathstrut \ell_3^j$}] (z'') at (0,0) {};
    \path[->] (z) edge (z');
    \path[->] (z') edge (z'');
	\node at (0,3) {\parbox{0.3\linewidth}{\subcaption{Graph of the clause agent $c_j$ for the $j$th clause $(\bar{x}_1 \vee x_2 \vee x_3)$}\label{subfig:agents bs}}};
    \end{scope}
    \end{tikzpicture}
    \caption{A schematic representation of the preference graphs of the different agents obtained in Theorem~\ref{fig:3sat construction}.}
    \label{fig:3sat construction}
\end{figure}
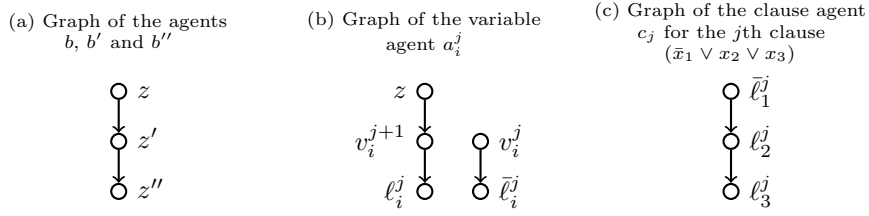
Note that we can obtain $I$ in polynomial time.

We claim that $\phi$ is satisfiable if and only if $I$ is satisfiable, that is, if there exists an allocation of the items in $I$ such
% so
that the dissatisfaction of each agent is at most $2$.

First, assume that $\phi$ is a satisfiable formula and denote by $\beta$ a satisfying assignment.
We construct an allocation $\pi$ for $I$ as follows.
We set $\pi(b) = \{z\}$, $\pi(b') = \{z'\}$ and $\pi(b'') = \{z''\}$.
Then for every variable $x_i$ that is assigned to \textsf{true} in $\beta$ we allocate the items $v_i^{j+1}$ and $\bar\ell_i^j$ to agent $a_i^j$ in $\pi$.
Conversely, for every variable $x_i$ that is assigned to \textsf{false} in $\beta$ we allocate items $v_i^j$ and $\ell_i^j$ to agent $a_i^j$.
Note that $\pi(a_i^j) \cap \pi(a_i^{j'}) = \emptyset$ if $j \neq j'$.
In other words, the agents $a_i^j$ and $a_i^{j'}$ receive different items in $\pi$.
Furthermore, all the agents $a_i^j$ have dissatisfaction $2$.
Note that every clause has at least one literal satisfying it.
In particular, for any $j \in \{1,\dots,m\}$, the $j$th clause contains a literal $u$ that satisfies it.
If $u=x_i$ for some $i \in \{1,\dots,n\}$, i.e., $u$ is a positive literal, then the variable $x_i$ is set to \textsf{true} in $\beta$, and thus $\pi(a_i^j) = \{v_i^{j+1},\bar\ell_i^j\}$.
This implies that the item $\ell_i^j$ is not allocated to any agent, and thus we can allocate $\ell_i^j$ to agent $c_j$.
Hence, agent $c_j$ has dissatisfaction at most $2$.
Similarly, if $u = \bar{x}_i$ for some $i \in \{1,\dots,n\}$, then the variable $x_i$ is set to \textsf{false} in $\beta$, and thus $\pi(a_i^j) = \{v_i^{j},\ell_i^j\}$.
So we can allocate $\bar\ell_i^j$ to agent $c_j$, who therefore has dissatisfaction at most $2$.
%%%
Thus, every clause agent $c_j$ can be assigned an item that has not been assigned yet, and we get that the dissatisfaction of $c_j$ is at most $2$.
Hence, $\pi$ is an allocation for $I$ in which all the agents have dissatisfaction at most $2$.

Now, assume that $I$ has an allocation $\pi$ for which all agents have dissatisfaction at most $2$.
We assume that $\pi$ does not allocate items to some agent if that agent has been allocated another preferred item (since removing the non-preferred item does not increase the dissatisfaction of the agent); we say that $\pi$ is \emph{economical}.
As the preference graph of every agent is a collection of at most $2$ paths, this implies that every agent is allocated at most $2$ items, and in particular at most $1$ item from each path in her preference graph.
Observe that, as in the previous case, each of the agents $b$, $b'$ and $b''$ have to be allocated exactly one of the items $z$, $z'$, or $z''$, as otherwise one of them would have dissatisfaction more than $2$.
Without loss of generality, we assume that $\pi(b) = \{z\}$, $\pi(b') = \{z'\}$, and $\pi(b'') = \{z''\}$.
 It follows that every variable agent $a_i^j$ has to be assigned at least $2$ items in her preference graph, as otherwise her dissatisfaction would be greater than $2$.
More specifically, the possible pairs of items that can be allocated to $a_i^j$ are $\{v_i^{j+1},v_i^j\}$, $\{v_i^{j+1},\bar\ell_i^j\}$ and $\{v_i^j,\ell_i^j\}$.
We claim that there is no variable agent $a_i^j$ such that $\pi(a_i^j) = \{v_i^{j+1},v_i^j\}$.
To get a contradiction, suppose that $\pi(a_i^j) = \{v_i^{j+1},v_i^j\}$.
Then $\pi(a_i^{j+1}) = \{v_i^{j+2},\bar\ell_i^{j+1}\}$, the only available pair for agent $a_i^{j+1}$ that has dissatisfaction at most $2$.
This in turn implies that $\pi(a_i^{j+2}) = \{v_i^{j+3},\bar\ell_i^{j+2}\}$, and more generally that all for all $j' \neq j$, $\pi(a_i^{j'}) = \{v_i^{j'+1}, \bar\ell_i^{j'}\}$.
However, agent $a_i^{j-1}$ cannot be allocated the items $v_i^{j}$ and $\bar\ell_i^{j-1}$ (with $j-1 = m$ if $j = 1$) since item $v_i^j$ is assigned to $a_i^j$.
Hence, for every variable agent $a_i^j$, the only possible pairs of items that can be allocated to $a_i^j$ are $\{v_i^{j+1},\bar\ell_i^j\}$ and $\{v_i^j,\ell_i^j\}$.

It is easy to notice that if $\pi(a_i^j) = \{v_i^{j+1},\bar\ell_i^j\}$ for some fixed $i \in \{1,\dots,n\}$ and $j \in \{1,\dots,m\}$, then in fact $\pi(a_i^{j'}) = \{v_i^{j'+1},\bar\ell_i^{j'}\}$ for all $j' \in \{1,\dots,m\}$, otherwise some of these agents will have dissatisfaction more than $2$.
We explain how to construct a satisfying assignment $\beta$ for $\phi$ using $\pi$.
For all $i\in \{1,\ldots, n\}$, if there exists a variable agent $a_i^j$ such that $\pi(a_i^j) = \{v_i^{j+1}, \bar\ell_i^j\}$, then we set $x_i$ to \textsf{true} in $\beta$.
Otherwise, we set $x_i$ to \textsf{false} in $\beta$.
It remains to show that $\beta$ is a satisfying assignment for $\phi$.
Fix $j \in \{1,\dots,m\}$ and consider the $j$th clause in $\phi$.
Recall that the preference graph of the clause agent $c_j$ is a path on three items, each of which representing the literals of the $j$th clause.
Since the agent $c_j$ has dissatisfaction at most $2$ and $\pi$ is economical, exactly one of those three items is assigned to $c_j$ in $\pi$.
Let $\tilde{\ell}_i^j$ be that item.
We have two cases to consider.
If $\tilde{\ell}_i^j = \ell_i^j$ (which corresponds to a positive literal), then $\pi(a_i^j) = \{v_i^{j+1}, \bar\ell_i^j\}$, and thus $x_i$ is set to \textsf{true} in $\beta$, which implies that the clause is satisfied.
Otherwise, if $\tilde{\ell}_i^j = \bar\ell_i^j$ (which corresponds to a positive literal), then $\pi(a_i^j) = \{v_i^{j}, \ell_i^j\}$, and thus $x_i$ is set to \textsf{false} in $\beta$, which implies that the clause is satisfied.
Since this holds for all $j \in \{1,\dots,m\}$, every clause in $\phi$ is satisfied by the assignment $\beta$.\qed
\end{proof}

\section{Parameterized algorithms}
\label{sec:junction}

Recall that in Section~\ref{sec:hardness} it was shown that both {\MINMAX} and
{\MINSUM} are computationally hard, even on special variants of out-trees. 
In the following we  show that  {\MINSUM} becomes polynomially solvable 
\newAD{if the arborization of the graphs is restricted,} namely by having only a constant number of \textit{junction vertices} (vertices with in- or out-degree greater than $1$). The corresponding Theorem~\ref{th:constjunc} also implies a polynomial algorithm for the case where all $G_i$ are directed matchings, but we already described a simpler approach in Theorem~\ref{th:minSUM-matching}.
It also gives polynomial algorithms for the case where $G_i$
are paths, which was solved by a straightforward method in Corollary~\ref{th:sum-paths}.

From Theorem~\ref{th:hard-A-matching} we know that {\MINMAX} remains NP-hard even for the special case where all $G_i$ are directed matchings, and thus do not contain any \newwSL{junction vertices} at all.
However, for a different setting without \newwSL{junction vertices} where all $G_i$ are paths, {\MINMAX} is  solvable in polynomial time (see Theorem~\ref{th:A-paths}). 

\smallskip
Let us now turn to {\MINSUM} and the above-mentioned restriction of the \newMM{preference graphs}. 
% expansion. 
Formally, we denote by $J_i \subseteq V(G_i)$ for each $i \in K$ the set of junction vertices in $G_i$, i.e., vertices in $G_i$ with in- or out-degree greater than $1$, and by $ \gamma = \sum_{i \in K} |J_i|$, the total number of \newMM{junction vertices} (counted with multiplicities).
Also, we call a vertex with in-degree $0$ and out-degree $1$ a \textit{simple source} and a vertex with in-degree $1$ and out-degree $0$ a \textit{simple sink}.
 Note that $\gamma$ constant implies that all $G_i$, except constantly many, consist only of collections of paths.
For background on fixed-parameter tractability, we refer to~\cite{MR3380745}.

\begin{theorem}\label{th:constjunc}
{\MINSUM} is fixed-parameter tractable with respect to $\gamma$.
\end{theorem}

\begin{proof}
We introduce an algorithm to solve the maximization problem for the total 
satisfaction, which implies the solution of  {\MINSUM}.
Note that there exists an optimal allocation $\pi$ which fulfills the \emph{minimality condition}, meaning it is minimal with respect to the property that for each agent $i$ no item allocated to agent $i$ is dominated by any other item allocated to $i$. 
Hence, we restrict our search and feasibility test to allocations fulfilling this condition. 

Given an allocation $\pi$ fulfilling the minimality condition, for each vertex $v \in J_i$ exactly one of the following four cases occurs.
% there exist only four possible pairwise-distinct cases.
\begin{enumerate}
	\item[(1)] $\pi$ allocates $v$ to agent $i$.
	\item[(2)] $\pi$ allocates some item in ${\it pred}_i(v)$ to $i$.
	\item[(3)] $\pi$ allocates some item in ${\it succ}_i(v)$ to $i$.
	\item[(4)] $\pi$ does not allocate any item in ${\it pred}_i(v) \cup {\it succ}_i(v) \cup \{v\}$ to agent $i$.
\end{enumerate}
In our algorithm, we enumerate all $4^{\gamma}$ possible assignments of
cases~(1)--(4) for all vertices in each $J_i$. For any such assignment we denote by 
$J_i^{(1)}, J_i^{(2)}, J_i^{(3)}, J_i^{(4)}$ the vertices in $J_i$ assigned to cases~(1)--(4).
\newSL{Note that if $v \in V$ is a \newwSL{junction vertex} in multiple agent graphs $G_i$ we enumerate all possible 
assignments for each of the agents independently. This is possible in time $4^{\gamma}$, since 
in the definition of $\gamma$ such \newwSL{junction vertices} are counted with multiplicity with respect to the 
graphs $G_i$.}
Then for each such assignment we test whether the assignment is feasible and if so find an optimal allocation of the remaining items subject to the conditions of cases~(1)--(4). 
Among all those allocations we take the one maximizing the total satisfaction.

Since not for all of the $4^{\gamma}$ assignments there exist feasible allocations $\pi$ fulfilling the minimality condition, we show how to test efficiently whether this is the case. 
First we check if some item $v$ is contained in $J_i^{(1)}$ and $J_{i'}^{(1)}$ for $i\neq i'$, which is clearly infeasible. 
In addition, for each agent $i$ we validate that no vertex in $J_{i}^{(1)}$ or $J_{i}^{(2)}$ is a predecessor of any vertex in $J_{i}^{(3)}$ and 
%no vertex in $J_{i}^{(4)}$ is predecessor or successor of any 
%vertex in $J_{i}^{(1)} \cup J_{i}^{(2)} \cup J_{i}^{(3)}$. 
no vertex in $J_{i}^{(4)}$ is predecessor of any vertex in $J_{i}^{(1)} \cup J_{i}^{(3)}$ 
and successor of any vertex in $J_{i}^{(1)} \cup J_{i}^{(2)}$.
Note that this can be easily done by running a breath- or depth-first-search in each graph $G_i$. 
Also note that each graph $G_i$ is a set of internally vertex-disjoint paths connecting the vertices in $J_i$ or paths connecting simple sources and simple sinks with each other or with vertices in $J_i$. 
What remains to be done is finding an optimal allocation of the items corresponding to the vertices on these paths that have not yet been allocated by any case~(1) assignment.

All the items corresponding to vertices in the subgraph of $G_i$ with a   
vertex in $J_i^{(4)}$ as their predecessor or successor cannot be allocated to agent $i$ and can 
also never contribute to the satisfaction of agent $i$.
Since we are restricted to allocations $\pi$ fulfilling the minimality condition, it holds that no item on a path entering a vertex in $J_{i}^{(1)}$ or $J_{i}^{(3)}$ can be allocated to agent $i$ and all these items cannot contribute to agent $i$'s satisfaction. 
Also, for all paths exiting vertices in $J_{i}^{(1)}$ or $J_i^{(2)}$, we already know that some predecessor of all these vertices is allocated.  
Hence, no item on these paths is allocated to agent $i$
by the minimality condition.

It remains to decide on the set of items on internal 
vertices of all paths connecting simple sources, simple sinks and $J_i^{(2)}$ and $J_i^{(3)}$ vertices.
For some of these paths we will have the condition that in a given subset of paths at least one vertex, i.e., the corresponding item, has to be allocated to agent $i$.
This aspect will be represented by defining for every agent $i$ a set $\mathcal{P}_{i}^{E}$ which is a subset of the power set of all these connecting paths.
Each of its elements is a collection of paths where at least one item of the internal vertices has to be allocated to agent $i$.
All other connecting paths in $G_i$ are optional and will be collected in the set of paths ${P}_i^{O}$.

The most involved setting concerns internal vertices on paths connecting vertices $J_i^{(3)}$ with vertices $J_i^{(2)}$.
Here, it is not directly clear on which of them it is mandatory to assign items to agent $i$ and on which it is just optional. 
We resolve this problem by guessing for each pair of vertices $v_1 \in J_i^{(3)}$, $v_2 \in J_i^{(2)}$ with at least one directed path traversing no \newwSL{junction vertex} between them whether at least one of the items on such a path must be assigned to agent $i$. 
If the guess chooses this option, we add the set of paths connecting $v_1$ to $v_2$ to $\mathcal{P}_{i}^{E}$,
otherwise the paths are added to ${P}_i^{O}$.
In the former case, also all other paths going from $v_1$ to simple sinks are optional (and added to ${P}_i^{O}$) as well as all paths coming into $v_2$ directly from simple sources.
If the guessing chooses not to make an item from any path between $v_1 \in J_i^{(3)}$ to an arbitrary other $v_2 \in J_i^{(2)}$ compulsory,
we can proceed with $v_1$ as with all other vertices in $J_i^{(3)}$: It follows that among all paths going from such a vertex to simple sinks at least one item must be allocated and thus the set of these paths is an element of $\mathcal{P}_{i}^{E}$. 
If no such paths exist the guesses are infeasible.
By the same reasoning, also for any $v_2 \in J_i^{(2)}$, if the guessing chooses not to make an item from any path between an arbitrary $v_1 \in J_i^{(3)}$ and $v_2$ compulsory, the set of all paths entering $v_2$ from some simple source is an element of $\mathcal{P}_{i}^{E}$.
Furthermore, there are paths connecting a simple source directly with a simple sink. 
These will be included in ${P}_i^{O}$.

Clearly, for every combination of guesses, we have to check whether the conditions $(2)$ and $(3)$ are  fulfilled for all vertices in $J_i^{(2)}$ and $J_i^{(3)}$. 
If this is the case, the decisions on the internal vertices will then be determined by the solution of the following max-profit flow problem.

\smallskip
\noindent\textsc{Max-Profit Flow}: 

\noindent \textbf{Input:} A  graph $H=(V(H), A(H))$, source $s \in V(H)$, sink $t \in V(H)$, lower and upper capacities $\ell,u \colon A(H) \rightarrow \mathbb{Z}_+$, a profit function $p \colon A(H) \rightarrow \mathbb{Z}_+$.

\noindent \textbf{Task:} Find an $s$,$t$-flow $f \colon A(H) \rightarrow \mathbb{Z}_+$ 
respecting the capacity bounds and maximizing the total profit $\sum_{e \in A(H)} p(e) f(e)$.

\smallskip
The max-profit flow problem can be reduced to the minimum-cost circulation problem and is thus solvable in strongly polynomial time (see, e.g.,~\cite[ch.~12]{MR1956924}).

In the next paragraphs we define an instance of the max-profit flow problem on a graph $H = (V(H),A(H))$.
The vertex set $V(H)$ consists of a source $s$, a sink $t$ and vertices $v_j$ for each item $j \in V$ that has not yet been allocated via an assignment of a vertex to $J_i^{(1)}$ for any $i \in K$.
Furthermore, $V(H)$ contains disjoint copies of all the internal vertices of paths in $P_{i}^{O}$ and paths in elements of $\mathcal{P}_{i}^{E}$ for all agents $i \in K$.  
In addition for each set of paths ${P}$ contained in $\mathcal{P}_{i}^{E}$ we add an auxiliary vertex $w_{{P}}$ to $V(H)$.
For all item vertices $v_j$ we add the arcs $(s,v_j)$ to $A(H)$.
The goal of the flow model is to represent the allocation of an item $j$ to some agent $i$ by one unit of flow going from $s$ to $v_j$, then to a vertex representing the item $j$ in $G_i$, 
further down to the end vertex of the respective path, and finally, possibly via an intermediate auxiliary vertex $w_{{P}}$, into $t$.
As a first step we set $\ell((s,v_j)) = 0$ and $u((s,v_j)) = 1$ to allow the allocation of each item $j$ at most once and $p((s,v_j)) = 0$.
Let $u_i \in V(H)$ be any vertex introduced as a copy of an internal vertex on any path in $G_i$ and let $j \in V$ be its corresponding item. 
Then we add the arc $(v_j,u_i)$ to $A(H)$ and set $\ell((v_j, u_i)) = 0$, $u((v_j, u_i)) = 1$ and $p((v_j,u_i)) = 1 + p_{u_i}$, where $p_{u_i}$ is the number of internal vertices in the path containing $u$ that are successors of $u_i$. 
Hence, sending flow along the arc $(v_j,u_i)$ corresponds to allocating $j$ to agent $i$.
%corresponding to the graph containing the path.
Moreover, $p((v_j, u_i))$ is exactly the extra satisfaction gained by this allocation.

For all paths in ${P}_{i}^{O}$ and all paths in elements of $\mathcal{P}_{i}^{E}$, we consider all arcs $(v', v'')$ in $G_i$ connecting two internal vertices $v'$, $v''$ of such a path.
For each such arc $(v', v'')$, we add an arc $a$ also in $A(H)$ as a connection between the two copies in $V(H)$ implied by the internal path vertices $v'$ and $v''$ in $V(G_i)$.
The parameters of each such arc $a$ are chosen as $\ell(a) = 0$, $u(a) = 1$ and $p(a) = 0$. 
In this way we ensure that at most one item in each path is allocated. 

For each set of paths ${P} \in \mathcal{P}_i^{E}$ 
we add in $A(H)$ an arc $a' = (v,w_{{P}})$ for the endvertex $v$ of every path in ${P}$.
We set $\ell(a') = 0$, $u(a') = 1$ and $p(a') = 0$. 
To enforce that at least one of the items corresponding to internal vertices of the paths in ${P}$ is allocated, we add arcs $a''=(w_{\mathcal{P}},t)$ to $A(H)$ and set $\ell(a'') = 1$, $u(a'') = \infty$ and $p(a'') = 0$. 
Finally, for all endvertices $v$ of paths in ${P}_{i}^{O}$ we add $(v,t)$ to $A(H)$ and set $\ell((v,t)) =0$ $u((v,t)) = \infty$ and $p((v,t)) = 0$.

It can be verified that any feasible flow is in one-to-one correspondence to an allocation of the internal vertices of the remaining paths fulfilling all the conditions stated above. 
The profit of such a flow corresponds exactly to the additional 
satisfaction gained by such an allocation. 
If no feasible flow exists, then some of the arcs with a lower capacity bound equal to $1$ could not be saturated with flow.
This implies that no feasible allocation exists which fulfills the required assignment for some internal vertices on certain sets of paths, although these were required by the current choice of the solution configuration.

We can compute such a max-profit flow in polynomial time. 
Adding the profit of this flow to the number of items assigned to any $J_i^{(1)}$ and all the satisfaction achieved through paths exiting vertices in $J_i^{(1)}$ or $J_i^{(3)}$ for all $i \in K$
gives the maximum possible total satisfaction for the current assignment. 
Taking the maximum over all $4^{\gamma}$ possible assignments, and considering for the remaining $\mathcal{O}(2^{(\gamma)^2})$ sets of paths all guesses whether an assignment is mandatory or optional, concludes the proof.
\qed
\end{proof}

\medskip
Following the hardness results in the preceding sections, one cannot expect to achieve efficient algorithms by just restricting graph parameters like the treewidth of the graphs $G_i$.
However, when we restrict the graph $G = (V,A)$ consisting of vertex set $V$ and arc set $A$, the union of the arc sets $A_i$ of all graphs $G_i$ for $i \in K$, there is still hope for positive results, as we explain next.
Classically, the notion of treewidth is defined for undirected graphs~\cite{robertson1984graph}. Yet, generalizations to directed graphs with vertex and edge labels have been defined for instance in~\cite{arnborg1991easy}. We make use of the representation of vertex- and arc-labeled digraphs using the following relational structure.

\begin{definition}[\cite{arnborg1991easy}]
    A \emph{labelled directed graph structure} is a tuple $(V,A,V_1, \dots,\allowbreak V_p, A_1, \dots, A_q)$, where
    % $V_1, \dots, V_p, A_1, \dots, A_q$ are unary predicates 
    % and $\textbf{head}, \textbf{tail}$ are binary predicates
    % such that
    \begin{itemize}
        \item $G = (V,A)$ is a graph,
        %\item $V$ designates the set of vertices,
        %\item $A$ designates the set of arcs,
        \item $V_1, \dots, V_p \subseteq V$ are special sets of vertices (labelled vertices),
        \item $A_1, \dots, A_q \subseteq A$ are special sets of arcs (labelled arcs).
    \end{itemize}
\end{definition}

Based on this definition we can now formally define the treewidth for such labelled graphs.

\begin{definition}[\cite{arnborg1991easy}]
    A \emph{tree decomposition} of a labelled directed graph structure $(V,A,V_1, \dots, V_p, A_1, \dots, A_q)$ is a pair $(T,\mathcal{S})$ where $T$ is a tree and $\mathcal{S}$ a family of sets indexed by the vertices of $T$, such that
    \begin{itemize}
        \item $\bigcup_{X_v \in \mathcal{S}} X_v = V \cup A$,
        \item for all $a = (u,w) \in A$ there exists a unique $X_v \in \mathcal{S}$ such that $a,u,w \in X_v$,
        \item for all $x \in V \cup A$ the subgraph of $T$ induced by $\{ v \colon x \in X_v \}$ is connected. 
    \end{itemize}
    The \textit{width} of such a tree decomposition is $\max_{X_v \in \mathcal{S}} |V \cap X_v| - 1$. The \textit{treewidth} of the labelled directed graph structure is the minimum width of a tree decomposition.
\end{definition}

Note that using this definition, the treewidth of a labelled directed graph structure $(V,A,V_1, \dots, V_p, A_1, \dots, A_q)$ is equal to the treewidth of the underlying undirected graph of $G = (V,A)$ using the classical definition of treewidth for undirected graphs. 
Also, given a labelled directed graph structure, a tree decomposition of width equal to its treewidth can be computed using a linear FPT algorithm with respect to the treewidth~\cite{MR1417901}.

In the following we analyze the complexity of {\MINSUM} and {\MINMAX} parameterized by the treewidth of the underlying undirected graph of $G = (V,A)$ and the number of agents $k$. 
We are able to exploit the treewidth since we are able to formulate feasible assignments of items and the corresponding sets of non-dominated items as a formula in the language of monadic second-order (MSO). This allows the application of variants of Courcelle's Theorem to obtain fixed parameter tractability results. For the definition of the MSO language and the related parameterized algorithms, which are heavily used in the following results, we refer the reader to~\cite{arnborg1991easy}.

\begin{sloppypar}
\begin{lemma}\label{lemma:msoformula}
Let $K = \{1,\dots,k\}$ be the set of agents with preference graphs $G_i = (V_i, A_i)$, and consider the labelled directed graph structure given by $(V,A,V_1, \dots, V_k, A_1, \dots, A_k)$, where $V = \bigcup_{i \in K} V_i$ and $A = \bigcup_{i \in K} A_i$.
For sets $\pi_1, \dots, \pi_{k}, U_1, \dots, U_{k}\subseteq V$, let  $\psi(\pi_1, \dots, \pi_{k}, U_1, \dots, U_{k})$ be the property that is defined to be true if and only if the mapping $\pi:K\to 2^V$ defined by $\pi(i) = \pi_i$ for all $i\in K$ is an allocation of items in $V$ to the $k$ agents and $U_i$ is the set of all items in $V_i$ that are not dominated by $\pi_i$ for each $i \in K$.
Then there exists an MSO formula expressing~$\psi$.
\end{lemma}
\end{sloppypar}

\begin{proof}
The property $\psi$ can be expressed in the MSO language as follows:
\begin{align*}
    \psi(\pi_1, & \dots , \pi_{k}, U_1, \dots, U_{k}) = \\
    & \bigwedge_{\substack{i, j \in K:\\ i \neq j}} (\forall v \in V \colon v \in \pi_i \rightarrow (v \in V_i \wedge v \notin \pi_j))  \\
    & \quad \wedge\\
    & \;\,\bigwedge_{i \in K} (\forall u \in V: u\in U_i \rightarrow u\in V_i)\\
    & \quad \wedge\\
    & \;\,\bigwedge_{i \in K} (\forall u \in V \colon  u \in U_i \leftrightarrow (\forall s \in \pi_i \colon \neg \mathbf{conn}_i(s,u))),
\end{align*}
where, $\mathbf{conn}_i(s,u)$ is the property that $s$ and $u$ are connected by 
a directed path of arcs in $G_i$, which can be expressed in the MSO language in the following way:
\[ \mathbf{conn}_i(s,u) = (\forall Y \subseteq V_i \colon (s \in Y \wedge u \notin Y) \rightarrow (\exists (v,w) \in A_i \colon v \in Y \wedge w \notin Y)). \]\qed
\end{proof}

\begin{sloppypar}
Using Lemma~\ref{lemma:msoformula} we are now ready to formulate the {\MINSUM} problem  as a linear EMSO optimization problem and reduce {\MINMAX} to a polynomial number of EMSO decision problems, implying a fixed parameter tractability result with respect to treewidth.
For the formal definitions of linear EMSO optimization problem and EMSO decision problem we again refer the reader to~\cite{arnborg1991easy}.
\end{sloppypar}

For simplicity, we assume in the next two theorems that each item $v\in V$ is desired by at least one agent $i\in K$.
This assumption is without loss of generality, since otherwise we can simply ignore the items in $V\setminus (\bigcup_{i\in K}V_i)$.
 
\begin{theorem}\label{th:tw-minsum}
The optimization variant of {\MINSUM} is fixed-parameter tractable with respect to $k+t$, where $k$ is the number of agents, and $t$ is the treewidth of the underlying undirected graph of $G = (V,A)$, where $A$ is the union of the arc sets $A_i$ of all graphs $G_i$ for $i \in K$.
\end{theorem}

\begin{proof}
We regard an instance $(V_i,A_i)$, $i \in K$, of the optimization variant of {\MINSUM} as a labelled directed graph structure, as in Lemma~\ref{lemma:msoformula}.
Following~\cite{MR1417901}, we can assume that an optimal tree decomposition of the labelled directed graph structure is available.
By \cite[Theorem~5.6]{arnborg1991easy} it thus suffices to show that {\MINSUM} can be formulated as a linear EMSO optimization problem, where the length of the used MSO formula is bounded by a function depending only on the number of agents $k$.

We consider the following optimization problem (\textsc{EMSO-Min-Sum}), which uses the property $\psi$ from Lemma~\ref{lemma:msoformula}.

\begin{align*}
    \min & \sum_{i \in K} |U_i|\\
    \text{s.t. } & \exists \pi_1, \dots, \pi_{k}, U_1, \dots, U_{k} \subseteq V \colon  G \models \psi(\pi_1, \dots, \pi_{k}, U_1, \dots, U_{k}).
\end{align*}

This problem is a linear EMSO optimization problem, since the objective function is linear in the cardinality of the free set variables of $\psi$ and $\psi$ is an MSO formula of length bounded by a function in $k$.
By the fact that for any agent $i \in K$ it holds that 
$U_i$ is the set of vertices not dominated by any item in $\pi_i$ and the sets $\pi_1, \dots \pi_k$ correspond to an allocation of items in $V$ to the $k$ agents, {\MINSUM} is equivalent to (\textsc{EMSO-Min-Sum}).
\qed
\end{proof}

\begin{theorem}\label{th:tw-minmax}
{\MINMAX} is fixed parameter tractable with respect to $d+k+t$, 
where $d$ is the dissatisfaction threshold, $k$ is the number of agents, and $t$ is the treewidth of the underlying undirected graph of $G = (V,A)$, where $A$ is the union of the arc sets $A_i$ of all graphs $G_i$ for $i \in K$.
\end{theorem}

\begin{proof}
To obtain the result, enumerate all the $\mathcal{O}(d^k)$ possible dissatisfaction profiles $(d_1, \dots, d_k)$ for the agents and check whether a solution giving exactly this dissatisfaction profile exists. Using the MSO formula $\psi$ from Lemma~\ref{lemma:msoformula} we do this by solving the following decision problem (\textsc{MSO-Dec}).
\begin{align*}
    \exists \pi_1, \dots, \pi_{k}, U_1, \dots, U_{k} \subseteq V \colon & G \models \psi(\pi_1, \dots, \pi_{k}, U_1, \dots, U_{k})\\
    & \forall i \in K \colon |U_i| = d_i.
\end{align*}
Note that this is an EMSO decision problem since we only added additional equality constraints on the cardinalities of the free set variables of the MSO formula $\psi$.
By \cite[Theorem~5.5]{arnborg1991easy} this implies the claimed result.
\qed
\end{proof}

\begin{remark}\label{rem:XP}
The running time of the parameterized algorithm given in the proof of Theorem~\ref{th:tw-minmax} is $\mathcal{O}(d^k f(t) n^{O(1)})$.
Since $d \leq n$, this is an XP algorithm for {\MINMAX} with respect to the parameter $k+t$.
\end{remark}

Observe that Theorem~\ref{th:A-matching-2} (the case of two agents with directed matchings as preference graphs) follows from Remark~\ref{rem:XP} by taking $k=t=2$. 
It remains open, however, if the above results can be generalized to the clique-width of $G$. 
Also, for {\MINMAX} it is unknown whether the problem is fixed-parameter tractable in $k+t$. We conjecture this not to be the case. Also the development of faster FPT algorithms for both {\MINSUM} and {\MINMAX} remain open problems for future research.

\section{Conclusion}
\label{sec:conc}

We have introduced a new model in which agents' preferences over indivisible items are captured by means of directed acyclic graphs (preference graphs). 
For this setting, we have analyzed the task of allocating items to agents in a way that minimizes either the total or the  maximum dissatisfaction.
The latter is measured by the number of desired items an agent does not receive and for which she does not get a more preferred item. 

Complementing our surprisingly strong hardness results we have presented several positive results, i.e., polynomial-time solvable cases.
We could also show that---from a complexity point of view---the min-max objective is sometimes harder than the min-sum objective.
Referring to the summary of our results in Tables~\ref{fig:overview} and~\ref{fig:overview-tw} (see the introduction), we 
gave a fairly complete characterization of the separation between NP-complete and polynomial cases with respect to the preference graphs' structure.
However, some interesting questions remain open. 
For instance, can we generalize  to more than two agents the positive results for {\MINSUM} for disjoint unions of out-stars (Theorem~\ref{th:B-stars-2}) and {\MINMAX} for directed matchings (Theorem~\ref{th:A-matching-2})?
More generally, which further graph structures admit positive results for our two objectives? 
And which additional parameters allow for fixed-parameter tractability (in particular, for {\MINMAX})?

\subsubsection*{Acknowledgements}

The authors wish to thank Matja\v{z} Krnc for valuable discussions. The work of this paper was done in the framework of two bilateral projects between University of Graz and University of Primorska, financed by the OeAD (SI 22/2018 and SI 31/2020) and the Slovenian Research Agency (BI-AT/18-19-005 and BI-AT/20-21-015).
The authors acknowledge partial support of the Slovenian Research Agency (I0-0035, research programs P1-0404, P1-0285, research projects N1-0102, N1-0160, N1-0210, J1-9110, and a Young Researchers Grant) and by the Field of Excellence ``COLIBRI'' at the University of Graz.

\bibliographystyle{plainnat}
\bibliography{refs}

%\printbibliography

\end{document}